\documentclass[12pt,oneside,english]{amsart}
\usepackage[utf8]{inputenc}
\usepackage[a4paper]{geometry}
\usepackage{footnote}
\usepackage{color}
\usepackage{babel}
\usepackage{array}
\usepackage{float}
\usepackage{url}
\usepackage[table]{xcolor}
\usepackage{multirow, hhline}
\usepackage[foot]{amsaddr}
\usepackage{caption}
\usepackage[normalem]{ulem}
\usepackage{enumitem}
\usepackage{multirow}
\usepackage{amsmath,amssymb,scalerel,stmaryrd}
\usepackage{amsthm}
\usepackage{graphicx}
\usepackage{setspace}
\usepackage[authoryear]{natbib}
\usepackage[unicode=true,pdfusetitle,
 bookmarks=true,bookmarksnumbered=false,bookmarksopen=false,
 breaklinks=false,pdfborder={0 0 1},backref=section,colorlinks=true]
 {hyperref}
\hypersetup{
 citecolor=cite-blue}
\usepackage{threeparttable}
\makeatletter

\providecommand{\tabularnewline}{\\}

  \theoremstyle{plain}
  \newtheorem{prop}{\protect\propositionname}

 \theoremstyle{plain}
  
  \newtheorem*{claim*}{\protect\claimname}

\theoremstyle{definition}
\newtheorem{example}{Example}[section]

  \theoremstyle{plain}
  \newtheorem{cor}{\protect\corollaryname}
  \theoremstyle{remark}
  \newtheorem{rem}{\protect\remarkname}
  \theoremstyle{definition}
  \newtheorem{defn}{\protect\definitionname}
\theoremstyle{plain}
\newtheorem{thm}{\protect\theoremname}
  \theoremstyle{plain}
\newtheorem*{lem*}{\protect\lemmaname}
\newtheorem{lem}{\protect\lemmaname}

\newtheorem{innercustomthm}{Theorem}
\newenvironment{customthm}[1]
  {\renewcommand\theinnercustomthm{#1}\innercustomthm}
  {\endinnercustomthm}
  
\newtheorem{innercustomprop}{Proposition}
\newenvironment{customprop}[1]
  {\renewcommand\theinnercustomprop{#1}\innercustomprop}
  {\endinnercustomprop}

\definecolor{cite-blue}{RGB}{0,0,204}
\date{}

\makeatother

  \providecommand{\definitionname}{Definition}
  \providecommand{\lemmaname}{Lemma}
  \providecommand{\propositionname}{Proposition}
  \providecommand{\remarkname}{Remark}
\providecommand{\corollaryname}{Corollary}
\providecommand{\theoremname}{Theorem}
\providecommand{\claimname}{Claim}

\newcommand{\stkout}[1]{\ifmmode\text{\sout{\ensuremath{#1}}}\else\sout{#1}\fi}

\DeclareUnicodeCharacter{2217}{*}
\begin{document}

\title{Pick-an-object Mechanisms}

\author{Inácio Bó}
\thanks{\textbf{Bó}: Southwestern University of Finance and Economics, China Center for Behavioral Economics and Finance, Chengdu, China; website: \protect\url{http://www.inaciobo.com}; e-mail:
mail@inaciobo.com}

\author{Rustamdjan Hakimov}
\thanks{\textbf{Hakimov}: Department of Economics, University of Lausanne, Switzerland and WZB Berlin Social Center, Germany.
e-mail: rustamdjan.hakimov@unil.ch}

\date{February, 2023}
\thanks{We thank Samson Alva, Lars Ehlers, Dorothea K\"ubler, Andrew Mackenzie, Alexander Nesterov, Marek Pycia, Madhav Raghavan, Peter Troyan, and Andrew Schotter for helpful comments. We thank Jennifer Rontganger for copy editing and Nina Bonge for support in running the experiments. Hakimov acknowledges financial support from the Swiss National Science Foundation (project \#100018\_189152).}

\maketitle

\begin{abstract}
We introduce a new family of mechanisms for one-sided matching markets, denoted pick-an-object (PAO) mechanisms. When implementing an allocation rule via PAO, agents are asked to pick an object from individualized menus. These choices may be rejected later on, and these agents are presented with new menus. When the procedure ends, agents are assigned the last object they picked. We characterize the allocation rules that can be sequentialized by PAO mechanisms, as well as the ones that can be implemented in a robust truthful equilibrium. We justify the use of PAO as opposed to direct mechanisms by showing that its equilibrium behavior is closely related to the one in obviously strategy-proof (OSP) mechanisms, but implements commonly used rules, such as Gale-Shapley DA and top trading cycles, which are not OSP-implementable. We run laboratory experiments comparing truthful behavior when using PAO, OSP, and direct mechanisms to implement different rules. These indicate that agents are more likely to behave in line with the theoretical prediction under PAO and OSP implementations than their direct counterparts.
\end{abstract}


\emph{\hspace{-0.6cm}Keywords}: Market Design, Matching, Sequential Mechanisms, Experiments, obvious strategy-proofness.

\section{Introduction\label{sec:Introduction}}

The literature of market design, and its applications, has grown and evolved greatly over recent years. Even if we restrict our attention to the design of centralized matching markets, the instances in which theoretical and empirical contributions have influenced the way resources are allocated seem to be continuously expanding. Examples include the design of school choice procedures \citep{Abdulkadiroglu2003-dx}, centralized college admissions \citep{Balinski1999-uj},  matching resident doctors to hospitals \citep{Roth1999-id}, organs to patients \citep{Roth2004-pw}, refugees to localities \citep{Jones2017-ae}, appointment slots for services \citep{hakimov2019avoid}, and more. By carefully choosing how to determine these allocations as a function of information, such as participants' preferences, priority structures, and fairness concerns, these procedures can lead to allocations that satisfy certain desirable properties.

One crucial challenge faced by the designer of these mechanisms is that the information needed to determine the desired allocation is known by the participants but not by the designer. These are, often, their preferences over outcomes, but may also include other relevant information, such as their socioeconomic status.\footnote{See \cite{Aygun2019-ko}.} In real-life applications, this issue is usually solved by combining two tools: dominant strategy implementation and the revelation principle. These guarantees that if the designer wants the strategic simplicity provided by dominant strategy implementation, it suffices to consider revelation mechanisms, in which participants are simply asked to report their private information, safe in the knowledge that they can be truthful. In fact, the vast majority of the literature focuses solely on direct revelation, and often strategy-proof, mechanisms.

However, recent experimental and field evidence has raised concerns about market participants' ability to understand these mechanisms' incentive properties. Many participants use dominated strategies, thus distorting the market allocations. This phenomenon was documented in many laboratory experiments\footnote{See \citet{hakimov2021experiments} for an extensive survey of the experimental matching literature.} and in papers using field data.\footnote{See our related literature section.} These present a new challenge: is there an alternative to the implementation of allocation rules via strategy-proof direct mechanisms that would result in behavior that is more often in line with the theoretical predictions?

One recent and celebrated attempt to formalize an alternative to strategy-proofness that accounts for the extent to which participants can easily understand the incentives induced by mechanisms was the concept of obvious strategy-proofness (OSP). ``A strategy is obviously dominant if, for any deviation, at any information set where both strategies first diverge, the best outcome under the deviation is no better than the worst outcome under the obvious dominant strategy, and a mechanism is obviously strategy-proof (OSP) if it has an equilibrium in obviously dominant strategies'' \citep{Li2017-oa}. OSP is, therefore, a refinement of the notion of strategy-proofness, in that obvious dominance implies weak dominance. This concept could help explain why, for example, laboratory experiments indicate that agents are more likely to bid truthfully in a clock auction than in a sealed-bid second-price auction \citep{Kagel1987-ii}. While both implement the same rule in truthful dominant strategies, the former is also obviously dominant as opposed to the latter.\footnote{The author also presents the results of laboratory experiments comparing behavior and outcomes under strategy-proof and OSP implementations of the random serial dictatorship rule, obtaining similar results. Note, however, that \citet{Breitmoser2019-ur} raise questions as to whether the difference can be attributed to obvious strategy-proofness.}

One important shortcoming of OSP, especially for practical purposes, is that it is a very restrictive concept. Rules that are commonly considered for object allocation problems, such as top trading cycles, and stable rules, cannot be implemented via OSP mechanisms \citep{Li2017-oa,Ashlagi2018-zw}. This, therefore, leaves a large set of problems without this kind of behavioral guidance. Another concern is raised in a recent paper by  \citet{Pycia2019-vl}, who characterize all OSP mechanisms in the domain of object allocation and show that for some OSP strategies participants must have extensive foresight to correctly predict feasible actions and outcomes, which can be congnitively demanding. To address this concern, the authors also define an even stronger concept, strong OSP, that does not require foresight from participants but leaves the designer with an essentially unique mechanism---sequential serial dictatorship (SD). In sequential SD, participants face a menu of objects in order of priority and simply pick their allocation. The strategy of choosing the most-preferred object on the menu is \textit{strongly obviously dominant}.  

Recent laboratory experiments indicate that the dynamic implementation of the deferred-acceptance rule (DA), in which the equilibrium behavior also consists of choosing the most-preferred object from menus, leads to higher rates of truthful behavior than its standard direct revelation counterpart \citep{Bo2020-jm,Klijn2019-bq}. The results are especially surprising given that truthful behavior is an equilibrium involving non-dominant strategies in dynamic DA, while direct DA is strategy-proof.

If one considers that the main driver behind the behavior that is more in line with the theory in OSP mechanisms is the fact that strategies are obviously dominant, then the forces behind the experimental results in \cite{Li2017-oa}, \cite{Bo2020-jm}, and \cite{Klijn2019-bq} would be unrelated, because in the latter the equilibrium strategy is not even dominant. In this paper, we conjecture that the main driver behind the observed behavior more in line with the theoretical predictions is the simple mechanics of the equilibrium strategy, in which agents ``pick'' the object they would like to have from a menu, as opposed to submitting a ranking of objects representing their preferences. This would provide a unified and alternative explanation to these experimental results. On this basis, we introduce a class of sequential revelation mechanisms that implement various object allocation rules via an equilibrium behavior with closely related mechanics. We denote them \textit{pick-an-object mechanisms} (PAO).

In a PAO mechanism, agents are asked to pick an object from individualized menus. These choices may be rejected later on, and agents are then presented with new menus containing strict subsets of the previous menus from which the previous choices have been redacted. When the procedure ends, agents are assigned the last object they picked, if any. A PAO mechanism ``sequentializes'' an allocation rule if it always results in the unique allocation consistent with preference profiles that could rationalize the choices made by the agents. Therefore, if agents simply choose their most-preferred object when given a menu, then the object they hold once the procedure ends is the one that the allocation rule determines given their true preferences. Notice, therefore, that truthful equilibrium behaviors in OSP and PAO mechanisms are closely related. While in the former, it can be expressed as ``Wait until you can pick your best feasible object,'' in the latter, it is ``Pick your best feasible object and wait to see if you can keep it.''

The simple mechanics involved in a PAO mechanism induces a trade-off: because information about an agent's preferences can only be obtained through choices from menus, obtaining more information on her preferences requires ruling out her last choice.  This, in turn, restricts the set of allocation rules that can be ``sequentialized'' via PAO mechanisms.\footnote{We provide a simple motivating example in section \ref{subsec:monotonicDisc}.} We characterize such rules (Theorem \ref{thm:PickAnObjectIFFTerminalAv}), via a new property that we denote \textit{monotonic discoverability}. Many familiar rules, such as the Gale-Shapley DA and top trading cycles, satisfy it (Proposition \ref{prop:generalizedDAMonDisc}). We characterize the rules that can be implemented in truthful strategies via a robust equilibrium (perfect ex-post equilibrium) as being those that are strategy-proof and satisfy monotonic discoverability (Theorem \ref{thm:PAOTruthfullIFSP+MD}). For the case of the Gale-Shapley DA rule or under an additional condition on the rules, we show that the truthful strategy might also characterize the unique equilibrium strategy profile (propositions \ref{prop:PAOUniqueRM} and \ref{prop:DAUnique}). These results come from the fact that every choice might be the last one---and therefore there is only one ``correct'' choice. Finally, we show that every non-bossy OSP-implementable rule is implementable in weakly dominant strategies via PAO mechanisms (Proposition \ref{prop:OSPImplementableViaPAO}). 

Our justification for the use of PAO mechanisms is non-standard, in the sense that the game-theoretical incentive properties of the PAO mechanisms are not stronger than the alternatives, that is, strategy-proof direct mechanisms and OSP mechanisms, but we test our conjecture via laboratory experiments. 

We test two allocation rules, the top trading cycles (TTC) and serial dictatorship (SD), and we construct three treatments for each: direct revelation implementation, PAO implementation, and OSP implementation. That is, for each one of these rules, we ran each of the three different mechanisms for implementing them. We changed the mechanisms implementing the rules between-subjects and the main goal of the experiment is to compare the performance of PAO mechanisms relative to the direct and OSP (when possible) counterparts in various environments (in our case, rules). We changed the allocation rules within-subjects out of practical considerations and because we are interested in comparative statics between the types of mechanisms, but not the rules. Because in our setup TTC is only OSP-implementable for certain ``acyclic'' priority structures \citep{Troyan2019-ah,mandal2022obviously}, we split the TTC environments into cyclic and acyclic priority structures, with the former having only direct and PAO implementations.

We find that, in fact, OSP implementations lead to the highest truth-telling rates across the board. When the OSP implementation exists (i.e., except for TTC with cyclic priority structures), OSP outperforms both PAO and direct implementations in terms of truthful preference revelation. When comparing the direct implementation of TTC vs. the PAO mechanism, the experiments show that the PAO implementation leads to a higher proportion of subjects following truthful equilibrium strategies for the TTC rule and no difference for the SD rule.\footnote{The experiments reported in \cite{Bo2020-jm} complement these with a comparison between the direct revelation Gale-Shapley DA with the iterative deferred acceptance mechanism, which is its PAO implementation. The results are in line with the ones that we present here: the PAO implementation of DA results in a higher proportion of truth-telling than its direct counterpart.} As for efficiency, PAO mechanisms lead to significantly higher efficiency than direct mechanisms for both TTC and SD. OSP mechanisms improve efficiency over the direct ones, but there is no significant difference to  PAO mechanisms. Thus, despite a higher proportion of truthful strategies, OSP mechanisms do not improve efficiency relative to PAO ones. This is because the deviations from the truthful strategy are more likely to be payoff-relevant in OSP than in PAO mechanisms.

When we look more closely at the results for the OSP implementation of TTC, however, we see a big difference in the rates of truthful behavior, depending on the nature of the obviously dominant strategy for a subject. We employ the characterization of  \citet{Pycia2019-vl}, who show that every OSP-implementable mechanism is equivalent to a ``millipede game.'' In millipede games, at every decision node each player has to choose between leaving with an object among those in a given menu (``clinching action''), and at most one ``passing action.'' Unlike the clinching actions---in which agents simply choose their allocation---the passing actions require foresight to correctly predict feasible options, which can be cognitively demanding. In our experiment, in the OSP implementation of the acyclic TTC rule, when the obviously dominant strategy consists of simply picking the most-preferred object---the ``clinching action''---the rate of truthful behavior is 93\%. But when it requires some degree of ``foresight,'' in that it involves the ``passing action,'' the rate of truthful behavior is only 56\%. This result strongly supports the strong OSP concept of  \citet{Pycia2019-vl}, and its myopic ``picking'' equilibrium as being a better predictor of behavior than OSP in general.

Thus far, we motivated the PAO mechanisms only by the attractiveness of straightforward strategies empirically, which was confirmed by our experiments. However, PAO mechanisms have other attractive features which, although not directly evaluated in this paper, are important in practical applications. First, they improve the acquisition of information on the options available during the execution of the procedure, by requiring coarser information on preferences and limiting the number of options available between steps \citep{grenet2019decentralizing,hakimov2021costly}. Second, by allowing the participants to 
\textit{experience} the steps involved in the production of the final allocation, they can be perceived as more transparent \citep{Hakimov2020-ig}. Third, they can make the equilibrium strategy simpler in markets with a very large number of options, such as nationwide college admissions,\footnote{During university admissions in China and Brazil, for example, students face thousands of programs and universities \citep{Gong2016-bn,Bo2016-id}.} when compared to direct mechanisms. This is because in most practical cases, designers constrain the length of a rank-ordered list, as ranking even 100 options seems to be a very hard task, while the number of choices that students have to make in a PAO implementation of DA, for example, is typically much smaller than the number of alternatives that should be ranked \citep{Bo2016-id}. This, of course, comes at the cost of needing a longer time for the mechanism to run, which is an important practical consideration.\footnote{In college admissions in France, which runs a mechanism where students dynamically receive offers from colleges, the deadline for decision ranges from 5 days at the start of the procedure to 1 day towards the end. The system has been in place since 2018.} Fourth, there has been increasing demand from policymakers for the use of dynamic college admission mechanisms, e.g., recent reforms of college admissions in France, Inner-Mongolia, Germany, and Tunisia \citep{Bo2016-id,Gong2016-bn,luflade2018value}. Our paper provides guidance for constructing dynamic versions of different allocation mechanisms while preserving good incentive properties.

To sum up, we show that the PAO environment has significant benefits over its direct counterpart in the general decision environment.  However, when there is the option to use OSP mechanisms, our experiments suggest using them. Because many allocation rules used in real-life allocation problems are not OSP-implementable, but are PAO-implementable, we interpret our experimental results as support for the choice of PAO mechanisms over their direct revelation counterparts in these cases.

\subsection*{Related Literature\label{subsec:Related-literature}}

In addition to the studies mentioned in the introduction, this paper is mainly related to two literature strands: the design of sequential allocation mechanisms and their theoretical properties and the behavioral and experimental aspects of market design. 

From the theoretical perspective, perhaps the closest paper to ours is \cite{Mackenzie2020-gv}. They consider the family of \textit{menu mechanisms}, which are also sequential revelation mechanisms in which participants are asked to choose from menus of possible assignments. As in the case of PAO mechanisms, they focus on those in which an agent can never select an assignment twice. Unlike PAO mechanisms, however, the definition of menu mechanisms does not imply a restriction on the set of allocation rules that can be sequentialized because an agent's assignment does not have to be the last choice of an agent and may, in fact, be an object that was not chosen at all. Despite considering this more general setup, they show that strategy-proof rules can be implemented in a truthful ex-post perfect equilibrium, a result analogous to our Theorem \ref{thm:PAOTruthfullIFSP+MD}.

Another closely related paper is \cite{Borgers2019-dp}. As in our case, they are concerned about the notion of simplicity in mechanisms. They define the class of ``strategically simple'' mechanisms, which are those in which an agent's optimal strategy depends only on first-order beliefs about preferences and rationality. Like the rules implemented in truthful equilibria in PAO mechanisms, these include all dominant strategy mechanisms and extend to others. It is worth noting, however, that PAO mechanisms are not necessarily strategically simple.

Other papers have also considered sequential versions of allocation rules, such as multi-unit auctions \citep{Ausubel2004-bc,Ausubel2006-oc}, stable matchings \citep{Bo2016-id,Kawase2019-gf,Haeringer2019-rl}, and more general allocations \citep{Schummer2017-hh}. Moreover, there is a growing literature evaluating the sequential mechanisms used in the field \citep{Gong2016-bn,grenet2019decentralizing,Veski2017-ix,Dur2018-cj}. \cite{milgrom1982theory} show, in the context of auctions with correlated values, that are potential advantages of using sequential mechanisms instead of direct ones. Other recent papers, such as \cite{Akbarpour2019-lg} and \cite{Hakimov2020-ig} show that the use of sequential mechanisms can also be explained by their transparency and credibility characteristics: the experience that participants have when interacting with these mechanisms can convey information that helps them ensure the allocation is produced by following the rules.

We also provide contributions to the literature documenting dominated behavior in matching mechanisms.  \cite{Shorrer2018-uk,Rees-Jones2018-uz,Hassidim2016-ls,chen2019self} and \cite{Artemov2017-jx} document dominated strategies being played in real-life centralized allocation processes. These are in line with laboratory experiments that also evaluate truthful behavior in strategy-proof mechanisms relative to dynamic mechanisms \citep{Echenique2016-cx,Bo2020-jm,Breitmoser2019-ur,Klijn2019-bq}, and truthful behavior in TTC and SD \citep{Chen2002-tf,Chen2006-lv,Pais2008-gq,Guillen2017-ft,Guillen2018-cg,Hakimov2018-qe}.

\section{Examples\label{sec:Example}}

Pick-an-object mechanisms are a class of sequential mechanisms where in the first period all participants are asked to choose an object from individualized menus. In every period that follows, one of the following takes place: (i) agents are assigned to the last object they picked (which may include the null option) and the procedure ends, or (ii) some agents are given new menus, which are strict subsets of the previous one, not containing the last choice. Variations on the contents of these menus, the identity of those who receive them and how their contents depend on agents' choices characterize the space of PAO mechanisms.

In this section, we show two simple examples of the use of PAO mechanisms to implement the Gale-Shapley Deferred Acceptance (DA) and Gale's Top Trading Cycle (TTC) rules.\footnote{These are classic mechanisms used in the matching literature, introduced in \cite{Gale1962-mn} and \cite{Shapley1974-ir}, respectively.} 

In both cases, there are three agents ($a_1$, $a_2$, and $a_3$) and four objects ($o_1$, $o_2$, $o_3$ and $o_4$). Agent $a_3$ has the highest priority at object $o_1$, followed by $a_1$. Agent $a_1$ has the highest priority at object $o_2$.

\setlength{\tabcolsep}{0.5em} 
{\renewcommand{\arraystretch}{1.2}
\begin{table}[t]
\begin{tabular}{ccccc}
                                              & \multicolumn{2}{c}{\textbf{TTC}} & \multicolumn{2}{c}{\textbf{DA}}                                                            \\ \cline{2-5}
\multicolumn{1}{r|}{\multirow{3}{*}{$t=1$}}   & $a_1$     & \multicolumn{1}{c|}{$\{\boxed{o_1},o_2,o_3,o_4\}$}    & $a_1$    & \multicolumn{1}{c|}{$\{\boxed{o_1},o_2,o_3,o_4\}$}   \\
\multicolumn{1}{l|}{}                         & $a_2$     & \multicolumn{1}{c|}{$\{\boxed{o_1},o_2,o_3,o_4\}$}    & $a_2$    & \multicolumn{1}{c|}{$\{\boxed{o_1},o_2,o_3,o_4\}$}   \\
\multicolumn{1}{l|}{}                         & $a_3$     & \multicolumn{1}{c|}{$\{o_1,\boxed{o_2},o_3,o_4\}$}    & $a_3$    & \multicolumn{1}{c|}{$\{o_1,\boxed{o_2},o_3,o_4\}$}   \\ \cline{2-5}
\multicolumn{1}{r|}{\multirow{3}{*}{$t=2$}}   & $a_1$     & \multicolumn{1}{c|}{$\{o_3,\boxed{o_4}\}$}    & $a_1$    & \multicolumn{1}{c|}{$\emptyset$}   \\
\multicolumn{1}{l|}{}                         & $a_2$     & \multicolumn{1}{c|}{$\emptyset$}    & $a_2$    & \multicolumn{1}{c|}{$\{\boxed{o_2},o_3,o_4\}$}   \\
\multicolumn{1}{l|}{}                         & $a_3$     & \multicolumn{1}{c|}{$\emptyset$}    & $a_3$    & \multicolumn{1}{c|}{$\emptyset$}   \\ \cline{2-5}
\multicolumn{1}{r|}{\multirow{3}{*}{$t=3$}}   &           & \multicolumn{1}{c|}{}                     & $a_1$    & \multicolumn{1}{c|}{$\emptyset$}                             \\
\multicolumn{1}{l|}{}                         &           & \multicolumn{1}{c|}{}                     & $a_2$    & \multicolumn{1}{c|}{$\emptyset$}                             \\
\multicolumn{1}{l|}{}                         &           & \multicolumn{1}{c|}{}                     & $a_3$    & \multicolumn{1}{c|}{$\{o_1,\boxed{o_3},o_4\}$}                             \\ \cline{2-5}
\multicolumn{1}{r|}{\multirow{3}{*}{Outcome}} & \multicolumn{2}{c|}{$a_1\rightarrow o_4$}                & \multicolumn{2}{c|}{$a_1\rightarrow o_1$}                                     \\
\multicolumn{1}{l|}{}                         & \multicolumn{2}{c|}{$a_2\rightarrow o_1$  }              & \multicolumn{2}{c|}{$a_2\rightarrow o_2$}                                     \\
\multicolumn{1}{l|}{}                         & \multicolumn{2}{c|}{$a_3\rightarrow o_2$}                & \multicolumn{2}{c|}{$a_3\rightarrow o_3$}                                     \\ \cline{2-5}

\end{tabular}
\caption{Examples of use of the canonical PAO mechanism for implementing the TTC and DA rules.\label{tab:ExampleTTCandDA}}
\end{table}

We will consider the use of the \emph{canonical} PAO mechanism for these two rules. In a canonical PAO mechanism, we interpret choices from menus as revealing information about preferences, and menus only contain objects that might be matched to the agent, conditional on the allocations that the rule maps to preferences consistent with these choices.\footnote{A more detailed and formal definition is given in section \ref{subsec:canonicalPAO}.} For both rules, therefore, all agents are given menus containing all objects in period $t=1$. Table \ref{tab:ExampleTTCandDA} summarizes the steps for each mechanism, showing the menus given to each agent, with the choice they make highlighted.

Consider first the TTC rule. At $t=1$, all agents are given menus containing all objects. Suppose that $a_1$ and $a_2$ choose $o_1$ and $a_3$ chooses $o_2$. These are interpreted as revealing their most preferred objects. Since $a_3$ has the highest priority at $o_1$ and $a_2$ has the highest priority in $o_2$, they form a trading cycle and the TTC rule will, for any such preference profile, match $a_2$ to $o_1$ and $a_3$ to $o_2$. The choice that $a_1$ made, therefore, is no longer feasible. In the next period she is given a menu with the remaining objects and will be matched to whatever object she chooses. In the example in the table, she chooses $o_4$ and the procedure ends with all agents matched to the last object they chose.

Next, consider the DA rule in which agents make the same choices in $t=1$.  Since $o_1$ was revealed to be the most preferred object by both $a_1$ and $a_2$ but $a_1$ has a higher priority than $a_2$, the DA rule guarantees that $a_2$ is not matched to $o_1$.\footnote{At this point, it is still possible that $o_1$ will end up matched to $a_3$, but it will definitely not be matched to $a_2$.} In period $t=2$, therefore, agent $a_2$ is given a menu containing all objects except $o_1$. If she chooses $o_3$ or $o_4$, the procedure would end. We will, however, consider the case in which she chooses $o_2$. Since she has the highest priority for that object, the DA rule determines that she will be matched to it. This leads to $a_3$ receiving a menu in $t=3$ containing all objects except $o_2$.\footnote{Notice that since she has the highest priority at $o_1$, this object remains in her menu.} In our example, she chooses $o_3$ and the procedure ends with all agents matched to the last object they chose.

\section{Model\label{sec:Model}}

Let $A=\left\{ a_{1},a_{2},\ldots,a_{n}\right\} $ be a finite set
of \textbf{agents}\footnote{To simplify notation, an agent $a_i$ might sometimes be denoted by her index $i$.} and $O=\left\{ o_{1},o_{2},\ldots,o_{m}\right\}\cup\left\{\emptyset\right\}$
be a set of \textbf{object types},\footnote{Whenever there is no ambiguity, we sometimes refer to them simply as objects.} where $\emptyset$ is the \textbf{null object}. Each agent $a$ has strict \textbf{preferences}
$P_{a}$ over the set $O$. Given $P_{a}$
we express the induced weak preference by $R_{a}$. That is, $oR_{a}o'$
if $oP_{a}o'$ or $o=o'$. We abuse notation and $P_{a}$ may represent
its binary relation (ex: $oP_{a}o'$) or a tuple of elements of $O$,
for example, $P_{a}=\left(o,o',\emptyset,o''\right)$, which implies
$oP_{a}o'P_{a}\emptyset P_{a}o''$. We also often treat tuples of distinct elements of $O$ as sets, if that does not create any ambiguity. For example, we may say that $\gamma=\left(o_1,o_2,o_3\right)$ and $\left\{o_2\right\}\subset \gamma$. Denote by $\mathbb{P}$ the set
of all strict preferences over $O$. An object $o\in O$ is \textbf{acceptable} to agent $a\in A$ if $o P_{a} \emptyset$, and \textbf{unacceptable} to $a$ if $\emptyset P_{a} o$.
A \textbf{preference profile} is a list $P=\left(P_{a_{1}},P_{a_{2}},\ldots,P_{a_{n}}\right)$. We denote by $P^{-a}$ the set of all preferences in $P$ except for $P_a$.
Let $\mathcal{P}=\mathbb{P}^{n}$ be the set of all preference profiles.
An \textbf{allocation} is a function $\mu:A\to O$.\footnote{Notice that while this is a model of discrete object allocation, there is no explicit notion of feasibility considering capacities. Feasibility is ``encoded'' in the allocation rules themselves: if an allocation is in the image of the rule, then it is feasible.} For a given allocation $\mu$, we say that \textbf{agent $a$'s assignment} under $\mu$ is $\mu(a)$. Let $\mathcal{M}$
be the set of all allocations. A random allocation is a probability distribution over $\mathcal{M}$.  A \textbf{rule} is a function $\varphi:\mathcal{P}\to\text{\ensuremath{\mathcal{M}}}$.
Denote by $\varphi_{a}\left(P\right)=\varphi\left(P\right)\left(a\right)$. A rule $\varphi$ is \textbf{individually rational} if, for any $P\in \mathcal{P}$ and $a\in A$, $\varphi_a\left(P\right)R_a\emptyset$. 

A rule is \textbf{strategy-proof} if for every agent $a\in A$, $P\in\mathcal{P}$, and ${P_a}'\in\mathbb{P}$, it is the case that 
 $\varphi_{a}\left(P_{a},P^{-a}\right)R_{a}\varphi_{a}\left({P_a}',P^{-a}\right)$.

Define a \textbf{choice history} $h$ as a sequence of tuples $\left( \left(\Omega_1,\omega_1\right), \left(\Omega_2,\omega_2\right), \ldots \right)$, where for every $i$, $\Omega_i\subseteq O$ and $\omega_i\in\Omega_i$. That is, a choice history is a sequence of sets of object types and elements of those sets. For example:

\[ \left( \left(\left\{o_1,o_2,o_3,o_4,o_5\right\},o_2\right), \left(\left\{o_3,o_4,o_5\right\},o_5\right) \right) \]

We denote by $\overrightarrow{h}$ the last choice in $h$. That is, if $h=\left( \left(\Omega_1,\omega_1\right) , \ldots, \left(\Omega_k,\omega_k\right) \right)$, $\overrightarrow{h}=\omega_k$. We say that $h$ is a \textbf{continuation history} of $h'$ if all the tuples in $h'$ are also in $h$. We denote by $H$ the set of all choice histories, including the empty choice history, represented by $h^{\emptyset}$. We say that a preference $P_a$ is \textbf{consistent with the choice history} $h$ if for every $\left(\Omega_i,\omega_i\right)\in h$ and $o\in \Omega_i$, $\omega_i R_a o$. We denote by $P\left(h\right)$ the set of all preferences that are consistent with $h$. We can also say that a choice history is consistent with a preference using inverse reasoning, and denote by $h\left(P_i\right)$ the set of all choice histories that are consistent with $P_i$.

Denote by \textbf{collective history} $h^A$ a list of $n$ choice histories: $\left(h_1,h_2,\ldots,h_n\right)$. We denote by $H^A$ the set of all collective histories, by $h^{A-\emptyset}$ the collective history consisting of $n$ empty choice histories, and by $h^A_i$ the i-th element in $h^A$. We also say that $h^A$ is a \textbf{continuation collective history} of ${h^A}'$ if each choice history in $h^A$ is a continuation history of its related history in ${h^A}'$.  A preference profile $P=\left(P_{a_1},P_{a_2},\ldots,P_{a_n}\right)$ is \textbf{consistent with the collective history} $h^A$ if for every $i$, $P_{a_i}$ is consistent with $h^A_i$.\footnote{Note that this implies that a preference profile is consistent with a collective history if \textbf{all} histories are consistent with some preference. In other words, it is necessary that all histories are rationalizable.} We denote by $P\left(h^A\right)$ the set of preference profiles that are consistent with $h^A$. We can also say that a collective history is consistent with a preference profile using inverse reasoning, and denote by $h^A\left(P\right)$ the set of all collective histories that are consistent with the preference profile $P$.

Next, we abuse notation and define $\varphi\left(h^A\right)$ as the allocations that are consistent with the application of the rule $\varphi$, given the information about the preference profile that can be deduced from taking the revealed preference approach to the collective history $h^A$:

\[\varphi\left(h^A\right)=\bigcup_{P\in P\left(h^A\right)} \varphi\left(P\right)\]

For a given collective history $h^A$, therefore, $\varphi\left(h^A\right)$ is a (possibly empty) subset of $\mathcal{M}$. 

We also define the set of \textbf{feasible assignments after $h^A$ for $a_i$}, or $\mu_i^\varphi\left(h^A\right)$ as:

\[\mu_i^\varphi\left(h^A\right)=\bigcup_{\mu\in \varphi\left(h^A\right)}\mu\left(a_i\right)\]

Let $\Phi=\left(2^O\right)^n$. That is, $\Phi$ is the set of n-tuples with subsets of $O$.  

A \textbf{menu function} $\mathbb{S}:H^A\to\Phi$ specifies, for each collective history, a list of menus to be given to the agents.

\begin{itemize}
    \item $\mathbb{S}\left(h^{A-\emptyset}\right)=\left(\phi^0_1,\phi^0_2,\ldots,\phi^0_n\right)$, where for every $i$, $\phi^0_i$ is a non-empty subset of $O$. These are the initial menus.
    \item For any collective history $h^A$ and agent $a_i$, let $k_i$ be the number of non-empty menus that $a_i$ faced in $h^A_i$. We have, therefore, that $h^A_i=\left( \left(\Omega_1^i,\omega_1^i\right), \left(\Omega_2^i,\omega_2^i\right), \ldots, \left(\Omega_{k_i}^i,\omega_{k_i}^i\right) \right)$. Then,  $\mathbb{S}\left(h^A\right)=\left(\phi_1,\phi_2,\ldots,\phi_n\right)$, where $\phi_i\subseteq \Omega_{k_i}^i\backslash\{\omega_{k_i}^i\}$.\footnote{The menu $\phi_i$ is, therefore, a subset of the last non-empty menu that $a_i$ received.} 
\end{itemize}

Denote by $\mathbb{S}^i(\cdot)$ the i-th element in $\mathbb{S}(\cdot)$ (or, when convenient, $\mathbb{S}^a(\cdot)$ to be the element in $\mathbb{S}(\cdot)$ associated with agent $a$). For a given menu function $\mathbb{S}$, we define the \textbf{pick-an-object mechanism $\mathbb{S}$} as follows:

\begin{itemize}
    \item Period $t=1$: For every agent $a_i$, ask her to choose one item in $\mathbb{S}^i\left(h^{A-\emptyset}\right)$. Let agent $a_i$'s choice be $\sigma_i^t$. Define $h^{A-1}$ to be the collective history such that for every agent $a_i$, $h^{A-1}_i=\left(\left( \phi^0_1, \sigma_i^1\right)\right)$.
    \item Period $t>1$: Let $\left(\phi^t_1,\phi^t_2,\ldots,\phi^t_n\right)=\mathbb{S}\left(h^{A-t}\right)$. 
    \begin{itemize}
        \item If for all $i$, $\phi^t_i=\emptyset$, then the procedure stops, and outputs the allocation $\mu$, where for each $i$, $\mu\left(a_i\right)=\overrightarrow{h^{A-t}_{i}}$.
        \item Otherwise, for every agent $a_i$, ask her to choose one item in $\phi^t_i$, if the menu is non-empty.\footnote{An agent receiving an empty menu represents a situation in which she is not called to make a choice. One could alternatively interpret agents receiving empty menus as ``inactive'' agents in this period.} Let agent $a_i$'s choice be $\sigma_i^t$. Define $h^{A-(t+1)}$ as the collective history such that for every agent $a_i$ who received a non-empty menu, $h^{A-(t+1)}_i= h_i^{A-t} \oplus \left( \phi^t_i, \sigma_i^t\right)$,\footnote{We use ``$\oplus$'' to denote concatenation. That is, for example, $\left(\left( \phi, \sigma\right), \left( \phi', \sigma'\right)\right)\oplus\left( \phi'', \sigma''\right)=\left(\left( \phi, \sigma\right),\left( \phi', \sigma'\right), \left( \phi'', \sigma''\right)\right)$.} and for those with an empty menu, $h^{A-(t+1)}_i=h^{A-t}_i$.
    \end{itemize}
\end{itemize}

Notice that because the menus given to each agent do not include her previous choices and never include objects that were not present in previous menus, every collective history that results from any number of periods of a PAO mechanism is consistent with a non-empty set of preference profiles. Moreover, since the agents' allocation must be their last choice and menus are subsets of previous ones, every feasible allocation for an agent is present in the first menu given at $t=1$. Given a pick-an-object mechanism $\mathbb{S}$, we can define the set $H^A_{\mathbb{S}}$ of \textbf{collective histories of $\mathbb{S}$}. In other words, $H^A_{\mathbb{S}}$ contains all the collective histories that can take place when the pick-an-object mechanism $\mathbb{S}$ is used. All of them start with every agent $i$ receiving a menu with $\mathbb{S}^i\left(h^{A-\emptyset}\right)$ and making a choice from that menu. The other menus that they may face, and in which order, are determined by all the possible combinations of choices that agents make and the definition of the pick-an-object mechanism $\mathbb{S}$.

When facing a PAO mechanism, a simple behavior that an agent may follow is what we call a \textbf{straightforward strategy}, which we define below.

\begin{defn}
An agent $a\in A$ follows a \textbf{straightforward strategy with respect to $P_a$} if, whenever presented with a menu $I\subseteq O$, she chooses the most-preferred element of $I$ according to $P_a$. 
\end{defn}

\begin{defn}
A pick-an-object mechanism $\mathbb{S}$ \textbf{sequentializes} the rule $\varphi$ if, for any preference profile $P$, the pick-an-object mechanism $\mathbb{S}$ provides menus such that when each agent $a\in A$ follows the straightforward strategy with respect to $P_{a}$, the outcome $\varphi\left(P\right)$ is produced. We say that there exists a pick-an-object mechanism that sequentializes some rule $\varphi$ if there exists a menu function $\mathbb{S}$ such that a pick-an-object mechanism $\mathbb{S}$ sequentializes $\varphi$.
\end{defn}

\subsection{Monotonic discoverability}
\label{subsec:monotonicDisc}

At first glance, it might seem like every rule can be sequentialized by some pick-an-object mechanism. After all, by asking agents to choose from menus, one can always recover as much information about their preferences as necessary to pinpoint an allocation for a given rule. The definition of PAO mechanisms, however, imposes some restrictions on the menus that can be presented to an agent and how that relates to her assignment. In particular, objects previously chosen cannot be in future menus, and the agent is assigned to the last object she chose. These conditions result in mechanisms that have a simple and intuitive operation, but also induce a trade-off between obtaining more information about her preference and the assignment that the rule determines for her.

To see this, consider a problem where $O=\left\{o_1,o_2,o_3,\emptyset\right\}$, $A=\left\{a\right\}$, and the rule $\varphi^*$ as follows:

  \begin{eqnarray*}
    o_1 P_a o_2 P_a o_3 & : & \mu(a)=o_1\\
    o_1 P_a' o_3 P_a' o_2 & : & \mu(a)=o_3\\
    \text{Otherwise} & : & \mu(a)=\emptyset
    \end{eqnarray*}

Suppose that we are looking for a PAO mechanism that sequentializes $\varphi^*$, and consider first which object types should be in the menu given to the agent in period $t=1$. Clearly, object types $o_1$ and $o_3$, as well as the null object $\emptyset$ must be in the menu, since the agent's potential assignments must be in it. So there are two possibilities for the first menu: $\phi=\left\{o_1,o_2,o_3,\emptyset\right\}$ and $\phi'=\left\{o_1,o_3,\emptyset\right\}$.

What happens if the agent chooses $o_1$ from the first menu? If the menu was $\phi$, this tells us that $o_1$ is the most-preferred object type. This, however, does not provide enough information to narrow down to a single allocation, since this choice is consistent with the preferences $P_a$  and $P_a'$ above, but also with, for example, $o_1 P_a''o_3 P_a'' \emptyset P_a'' o_2$. More information about the agent's preferences is necessary to know which allocation to produce. 

More information on the agent's preferences can only be obtained by observing her choice from another menu. As per the definition of PAO mechanisms, that new menu and all the following ones cannot have $o_1$ as one of the options. We therefore reach an \textbf{informational gridlock}: in order to know which assignment to produce, we need to know, at the very least, whether $o_2$ is preferred to $o_3$ or not. The ``costs'' of obtaining more information  is the elimination of $o_1$ as a possible assignment for $a$. But if her preference is $P_a$, the rule $\varphi^*$ indicates that that should be her assignment. The same problem is present if the first menu was $\phi'$: if she chooses $o_1$ or $o_3$, we reach the same gridlock.\footnote{Note that this observation does not rely on the fact that the definition of PAO mechanisms imply that the first menu must contain all the agent's feasible assignments. The type of ``informational gridlock'' is obtained if we simply require the agent to be assigned her last choice and menus to be such that the choices are always rationalizable.}

 We conclude, therefore, that $\varphi^*$ cannot be sequentialized by a PAO mechanism. This is because $\varphi^*$ does not satisfy \textit{monotonic discoverability}, a property that we now define.

Let $\mu$ be an allocation. 
Let $P=\left(P_{a_{1}},P_{a_{2}},\ldots,P_{a_{n}}\right)$ be any
element of $\mathcal{P}$. We define the lower contour set of $P$ at $\mu$:
\vspace{-3pt}
\[\mathcal{L}\left(P,\mu\right)=\left\{ P'\in\mathcal{P}:\forall a\in A\text{ and }o,o'\in O:oP_{a}o'P_{a}\mu\left(a\right)\iff oP'_{a}o'P'_{a}\mu\left(a\right)\right\} \]

That is, $\mathcal{L}\left(P,\mu\right)$ contains all preference profiles that, for each agent $a$, agree with $P_a$ with respect to $\mu(a)$ and all objects preferred by $a$ to $\mu(a)$, but may differ with respect to objects that $\mu(a)$ are preferred to, with respect to $P_a$. We will say that $\mathcal{L}\left(P,\mu\right)$ is therefore the \textit{lower contour set} of $P$ at $\mu$. We will denote each element of $\mathcal{L}\left(P,\mu\right)$ a \textbf{continuation profile of $P$ at $\mu$}.

\begin{defn}
A rule $\varphi$ satisfies \textbf{monotonic discoverability} if, for
any allocation $\mu$ and preference profile $P$, either $\varphi\left(P\right)=\mu$
 or there is an agent $a^{*}\in A$ such that for all $P'\in\mathcal{L}\left(P,\mu\right)$, $\mu\left(a^{*}\right)\neq\varphi_{a^{*}}\left(P'\right)$.
\end{defn}

In words, when a rule satisfies monotonic discoverability, given any preference profile $P$ and allocation $\mu$, either the rule maps $P$ to $\mu$, or there is an agent $a^*$ who is not matched to her assignment under $\mu$ for any preference profile in the lower contour set of $P$ at $\mu$.

\begin{figure}[tb]
    \begin{minipage}[b]{.5\linewidth}
    \renewcommand{\arraystretch}{1.3}
      \centering
      \caption*{Step 1}
      \vspace{0.2cm}
        \begin{tabular}{l|l|l|l|l|l|l|}
            
            \hhline{~------}
            $P_{a_1}$ & \cellcolor{blue!25} $o_1$ & \cellcolor{black!25}$o_2$ &  \cellcolor{black!25}$o_3$  & \cellcolor{black!25}$o_4$  & \cellcolor{black!25} $o_5$ & \cellcolor{black!25}$\emptyset $ \\ \hhline{~------}
            $P_{a_2}$ & \cellcolor{blue!25} $o_3$ & \cellcolor{black!25}$o_1$ &  \cellcolor{black!25}$o_4$  & \cellcolor{black!25}$o_3$  & \cellcolor{black!25} $o_5$ & \cellcolor{black!25}$\emptyset $ \\ \hhline{~------}
            $P_{a_3}$ & \cellcolor{blue!25} $o_2$ & \cellcolor{black!25}$o_5$ &  \cellcolor{black!25}$o_4$  & \cellcolor{black!25}$o_1$  & \cellcolor{black!25} $o_3$ & \cellcolor{black!25}$\emptyset $ \\ \hhline{~------}
            $P_{a_4}$ & \cellcolor{blue!25} $o_5$ & \cellcolor{black!25}$o_4$ &  \cellcolor{black!25}$o_2$  & \cellcolor{black!25}$\emptyset $  & \cellcolor{black!25} $o_1$ & \cellcolor{black!25}$o_3$ \\ \hhline{~------}
        \end{tabular}
        
    \end{minipage}%
    \begin{minipage}[b]{.5\linewidth}
    \renewcommand{\arraystretch}{1.3}
      \centering
      \caption*{Step 2}
      \vspace{0.2cm}
        \begin{tabular}{l|l|l|l|l|l|l|}
            \hhline{~------}
            $P_{a_1}$ & \cellcolor{blue!25} $o_1$ & \cellcolor{black!25}$o_2$ &  \cellcolor{black!25}$o_3$  & \cellcolor{black!25}$o_4$  & \cellcolor{black!25} $o_5$ & \cellcolor{black!25}$\emptyset $ \\ \hhline{~------}
            $P_{a_2}$ & \cellcolor{red!25} \stkout{$o_3$} & \cellcolor{blue!25}$o_1$ &  \cellcolor{black!25}$o_4$  & \cellcolor{black!25}$o_3$  & \cellcolor{black!25} $o_5$ & \cellcolor{black!25}$\emptyset $ \\ \hhline{~------}
            $P_{a_3}$ & \cellcolor{blue!25} $o_2$ & \cellcolor{black!25}$o_5$ &  \cellcolor{black!25}$o_4$  & \cellcolor{black!25}$o_1$  & \cellcolor{black!25} $o_3$ & \cellcolor{black!25}$\emptyset $ \\ \hhline{~------}
            $P_{a_4}$ & \cellcolor{blue!25} $o_5$ & \cellcolor{black!25}$o_4$ &  \cellcolor{black!25}$o_2$  & \cellcolor{black!25}$\emptyset $  & \cellcolor{black!25} $o_1$ & \cellcolor{black!25}$o_3$ \\ \hhline{~------}
        \end{tabular}
    \end{minipage}
    
    \begin{minipage}[b]{.5\linewidth}
    \renewcommand{\arraystretch}{1.3}
      \centering
      \caption*{Step 3}
      \vspace{0.2cm}
         \begin{tabular}{l|l|l|l|l|l|l|}
            \hhline{~------}
            $P_{a_1}$ & \cellcolor{red!25} \stkout{$o_1$} & \cellcolor{blue!25}$o_2$ &  \cellcolor{black!25}$o_3$  & \cellcolor{black!25}$o_4$  & \cellcolor{black!25} $o_5$ & \cellcolor{black!25}$\emptyset $ \\ \hhline{~------}
            $P_{a_2}$ & \cellcolor{red!25} \stkout{$o_3$} & \cellcolor{blue!25}$o_1$ &  \cellcolor{black!25}$o_4$  & \cellcolor{black!25}$o_3$  & \cellcolor{black!25} $o_5$ & \cellcolor{black!25}$\emptyset $ \\ \hhline{~------}
            $P_{a_3}$ & \cellcolor{blue!25} $o_2$ & \cellcolor{black!25}$o_5$ &  \cellcolor{black!25}$o_4$  & \cellcolor{black!25}$o_1$  & \cellcolor{black!25} $o_3$ & \cellcolor{black!25}$\emptyset $ \\ \hhline{~------}
            $P_{a_4}$ & \cellcolor{red!25} \stkout{$o_5$} & \cellcolor{blue!25}$o_4$ &  \cellcolor{black!25}$o_2$  & \cellcolor{black!25}$\emptyset $  & \cellcolor{black!25} $o_1$ & \cellcolor{black!25}$o_3$ \\ \hhline{~------}
        \end{tabular}
    \end{minipage}%
    \begin{minipage}[b]{.5\linewidth}
    \renewcommand{\arraystretch}{1.3}
      \centering
      \caption*{Step 4}
      \vspace{0.2cm}
        \begin{tabular}{l|l|l|l|l|l|l|}
            \hhline{~------}
            $P_{a_1}$ & \cellcolor{red!25} \stkout{$o_1$} & \cellcolor{green!25}$o_2$ &  \cellcolor{black!25}$o_3$  & \cellcolor{black!25}$o_4$  & \cellcolor{black!25} $o_5$ & \cellcolor{black!25}$\emptyset $ \\ \hhline{~------}
            $P_{a_2}$ & \cellcolor{red!25} \stkout{$o_3$} & \cellcolor{green!25}$o_1$ &  \cellcolor{black!25}$o_4$  & \cellcolor{black!25}$o_3$  & \cellcolor{black!25} $o_5$ & \cellcolor{black!25}$\emptyset $ \\ \hhline{~------}
            $P_{a_3}$ & \cellcolor{green!25} $o_2$ & \cellcolor{black!25}$o_5$ &  \cellcolor{black!25}$o_4$  & \cellcolor{black!25}$o_1$  & \cellcolor{black!25} $o_3$ & \cellcolor{black!25}$\emptyset $ \\ \hhline{~------}
            $P_{a_4}$ & \cellcolor{red!25} \stkout{$o_5$} & \cellcolor{red!25}\stkout{$o_4$} &  \cellcolor{green!25}$o_2$  & \cellcolor{black!25}$\emptyset $  & \cellcolor{black!25} $o_1$ & \cellcolor{black!25}$o_3$ \\ \hhline{~------}
        \end{tabular}
    \end{minipage}
    \vspace{0.4cm}
    \caption{Monotonic discoverability and sequentialization}
\label{fig:monotonicDiscSequentialization}
\end{figure}

To understand the critical role that monotonic discoverability has when sequentializing rules, consider Figure \ref{fig:monotonicDiscSequentialization}. Suppose that we have a rule $\varphi$ that satisfies monotonic discoverability, and start with the allocation $\mu^0$ that matches each agent with her most-preferred object type, highlighted in blue in Figure \ref{fig:monotonicDiscSequentialization}, step 1. We can consider two cases. In the first, $\varphi$ is such that for every preference profile in which agent $a_i$'s top option is as in $P$, these agents are matched to their top choices. In that case, knowing those top choices already gives us enough information to determine this to be the outcome that $\varphi$ maps for any continuation profile.\footnote{One simple example of this situation is when $\varphi$ is a simple serial dictatorship, and no two agents have the same object as their most-preferred one, as in $P$.} The second case is when this is not true. That is, there are continuation profiles of $P$ at $\mu^0$ for which $\varphi$ does not map those profiles to $\mu^0$. Let $P^0$ be one such profile. By monotonic discoverability, there is at least one agent who, for every continuation profile of $P^0$ at $\mu^0$ (and therefore also of $P$ at $\mu^0$) will not be matched to her match at $\mu^0$. Without loss of generality, let $a_2$ be such an agent.

Consider next the allocation $\mu^1$, which is the same as $\mu^0$ except that the object mapped to $a_2$ is the second one in her preference. This is shown in Figure \ref{fig:monotonicDiscSequentialization}, step 2. Here, once again, we have two cases. In the first, $\varphi$ is such that for every continuation preference profile of $P$ at $\mu^1$, all agents are matched to their outcomes under $\mu^1$. The second case is where there is at least one, but potentially many, such profiles in which the outcomes are different than $\mu^1$. There is one thing we can say, however. Because the continuation profiles of $P$ at $\mu^1$ are also continuation profiles of $P$ at $\mu^0$, for none of these cases is agent $a_2$ matched to $o_3$. Suppose here, without loss of generality, that the second case is true, agents $a_1$ and $a_4$ are not matched to their outcomes under $\mu^1$ for any continuation profile of $P$ at $\mu^1$. Steps 3 and 4 represent a continuation of this argument, until at step 4, the allocation $\mu^4$ being evaluated is in fact the one that the rule $\varphi$ maps to all continuation profiles of $P$ at $\mu^4$. Notice that this process always converges to the allocation mapped by $\varphi$ to the profile $P$, because we only move down the preference ordering of an agent when we have enough information to determine that $\varphi$ rules out previously considered matches to her.

Following this process shows us two consequences of a rule satisfying monotonic discoverability. The first is that while following this monotonic process of evaluating allocations that are at each step weakly worse (from the perspective of the underlying preference profile), we end up at the allocation that $\varphi$ maps to $P$. This fact is fundamental for a PAO mechanism to be able to match agents to their last choices without repetition of their choices. 

The second is that monotonic discoverability solves the informational gridlock that is induced by a PAO mechanism; that is, we must permanently reject the last choice the agent made to obtain more information about her preferences. Monotonic discoverability guarantees that whenever we do not have enough information to single out an allocation, there is at least one agent whose last choice can be rejected and then presented with a new menu. This allow us to obtain more information about preferences, narrowing down the set of preference profiles consistent with the agents' choices. In fact, this relationship between monotonic discoverability and PAO mechanisms is as strong as possible, as shown in Theorem 1 below.

\begin{thm}
\label{thm:PickAnObjectIFFTerminalAv}
There exists a pick-an-object mechanism that \textbf{sequentializes} an individually rational rule $\varphi$ if and only if $\varphi$ satisfies monotonic discoverability.
\end{thm}

Theorem \ref{thm:PickAnObjectIFFTerminalAv} implies, therefore, that the structure of PAO mechanisms limits the rules that can be sequentialized using them. If we were to use choices from menus simply as means to obtain information on the preferences, we would have no such restriction. One could, for example, make each agent choose one object at a time starting from the entire set $O$, eliciting the entire preference profile, which could then be used to determine the allocation. The participant's experience with this type of mechanism would be fundamentally different, though: the connection between their choices and assignments would be unclear, in that they might end up matched to an object chosen early in the process, before having seen many other choices. In PAO mechanisms, there is a much more explicit connection between choices and assignments: an agent can pick what she wants and will be able to keep it unless the information provided by the other participants, through their choices, determines that she cannot keep it, in which case she will be offered the chance to choose again from a new menu with feasible options.\footnote{The specific implementation of the PAO mechanism could even allow for the market designer to credibly and truthfully explain the reason why the agent cannot keep the object, using the information that she obtained that resulted in that rejection.}

\subsection{Generalized Deferred Acceptance Procedures}\label{subsec:generalizedDAMechanisms}

Many mechanisms used in matching are defined by algorithms that produce outcomes, as opposed to axioms or objective functions. Two classic examples are the Gale-Shapley deferred acceptance mechanism (DA) \citep{Gale1962-mn} and Gale's top trading cycles (TTC) \citep{Shapley1974-ir}. Their definitions describe step-by-step procedures that use agents' preference rankings (and often other information, such as priority orderings) and result in an allocation.

\textit{Generalized DA procedures} is a general description of those which produce tentative allocations at each step, following the agent's preferences, until no agent has her choice rejected. They generalize DA in the sense that they determine whether the tentative allocation of an agent to an object type becomes a rejection based on the \textit{entire} tentative allocation and set of proposals, as opposed to the tentative allocation and proposals \textit{to the object type in question}. If we use the college admissions analogy, in a generalized DA procedure, whether a student's proposal is accepted or not may depend not only on the college she applies to (and its tentatively matched students) but also on the entire tentative assignments of students to colleges, and contemporaneous applications.

A generalized DA procedure is, therefore, defined by an \textit{update function} $\Psi:\mathcal{M}\times \mathcal{M}\to \mathcal{M}$. In it, $\Psi\left(\mu,\mu'\right)$ informs, for a ``tentative'' assignment $\mu$, what will be the new tentative assignment when the new proposals are the ones represented by the assignment $\mu'$. This has the restriction that, if $\Psi\left(\mu,\mu'\right)=\mu''$, for every $a\in A$ it must be that $\mu''(a)\in \{\emptyset, \mu(a),\mu'(a)\}$.\footnote{Interestingly, the definition of generalized DA procedures coincides with a deferred acceptance procedure used by \cite{pycia2021matching} to produce stable matchings in the presence of externalities. While this points to a relation between the family of rules that we are defining and their outcomes potentially satisfying certain properties in the presence of those externalities, these are not the subject of our analysis, neither are sequential mechanisms the subject of theirs.} That is, the new tentative allocation must be one that for each agent, either leaves her with the same tentative allocation, unmatched, or tentatively assigned to the object type she has just proposed. Given an update function $\Psi$, the procedure can be described by the following algorithm:

\begin{itemize}
    \item \textbf{Step $1$}: Let $\mu^0$ be an assignment in which, for all $a\in A$, $\mu^0(a)=\emptyset$, and $\mu^*$ be an assignment consisting of all agents matched to their top choice in $P$.  Moreover, let $\mu^1=\Psi\left(\mu^0,\mu^*\right)$. We say that every agent $a\in A$ for which $\mu^*(a)\neq \emptyset$ and $\mu^1(a)= \emptyset$ was \textit{rejected by} $\mu^*(a)$.
    \item \textbf{Step $t>1$}: Construct the assignment $\mu^*$, in which every agent $a$ who was rejected at step $t-1$ is matched to her most-preferred object, with respect to $P_a$, from which she was not previously rejected. Moreover, let all other agents be matched to $\emptyset$ in $\mu^*$. Moreover, let $\mu^{t}=\Psi\left(\mu^{t-1},\mu^*\right)$. If for every $a\in A$ it is the case that $\mu^{t}(a)\in \left\{\mu^{t-1}(a),\mu^*(a)\right\}$, stop the procedure and determine the assignment to be $\mu^{t}$. Otherwise, go to step $t+1$.
\end{itemize}

Rules that are described by generalized DA procedures satisfy monotonic discoverability, as shown below.

\begin{prop}
\label{prop:generalizedDAMonDisc}
If $\varphi$ is described by a generalized DA procedure, then $\varphi$ satisfies monotonic discoverability.
\end{prop}

Proposition \ref{prop:generalizedDAMonDisc} implies that many mechanisms that are commonly used and referenced in the matching literature satisfy monotonic discoverability and can therefore be sequentialized by a PAO mechanism. One can easily see that these include, for example, DA itself, TTC, and the Boston mechanism \citep{Abdulkadiroglu2003-dx}. Note, however, that the converse does not hold, as shown in the example below.

\begin{example}[\emph{A rule that is monotonic discoverable but not generalized DA}]\label{ex:MonDiscButNotGenDA}
 Let $O=\{o_1,o_2,o_3\}$, $\varphi$ be defined for a single agent $a$,  and assign her $o_1$ whenever $o_1$ is her most-preferred object, $o_2$ when  $o_3 P_a o_1 P_a o_2$, and leave the agent unmatched otherwise.
\end{example}
  
Example \ref{ex:MonDiscButNotGenDA} involves a situation that cannot be handled by a Generalized DA: the mechanism must ``remember'' that the agent first chose $o_3$ (and was rejected) before choosing $o_1$ in order for the outcome to be $o_2$. Monotonic discoverability does not impose this ``short memory'' assumption into the sequentialization of the rule.

\subsection{A canonical Pick-an-object mechanism}
\label{subsec:canonicalPAO}

Up to this point, we have described how PAO mechanisms operate, and which rules can be sequentialized with PAO mechanisms. The next natural question is: given some rule $\varphi$ that satisfies monotonic discoverability, how can we design the contents of the menus given to the agents in each step, in a way that sequentializes $\varphi$?

In principle, there might be multiple PAO mechanisms that sequentialize a given rule, with variations over which agents are asked to choose from menus and the contents of those menus. We can, however, construct a \textbf{canonical PAO mechanism} for that rule. Let $\varphi$ be a rule that satisfies monotonic discoverability, and define the PAO function $\mathbb{S}$, such that for every $i$ and $h^A\in H^A$:

\[\mathbb{S}^i\left(h^A\right)=
\begin{cases}
    \emptyset & \text{ if } \left|\varphi\left(h^A\right)\right|=1 \text{ or } \overrightarrow{h^A_i}\in \mu_i^\varphi\left(h^A\right) \\
    \mu_i^\varphi\left(h^A\right) & \text{otherwise}
\end{cases}
\]

The proof of Theorem \ref{thm:PickAnObjectIFFTerminalAv} involves showing that the PAO mechanism $\mathbb{S}$ above sequentializes the rule $\varphi$. The resulting PAO mechanism can also be easily explained:

\begin{itemize}
    \item Step 1: Every agent $i$ is given a menu containing $\mu_i^{\varphi}\left(h^{\emptyset}\right)$, that is, all the object types that are allotted to $i$ in some allocation produced by the rule $\varphi$.\footnote{In other words, the set of feasible assignments for $i$ under $\varphi$. In a college admissions environment, for example, that would be every college that deems $i$ acceptable.}
    \item Step $t>1$: Let $h^A$ be the collective history representing the menus and choices made by the agents up to period $t$. There are two cases.
        \begin{itemize}
            \item $\left|\varphi\left(h^A\right)\right|=1$, that is, every preference profile consistent with the collective history up to step $t$ is mapped to the same assignment by $\varphi$. In that case, by monotonic discoverability, the last object chosen by the agents is exactly that assignment, and therefore the procedure ends and agents leave with the last object they picked.
            \item Otherwise, monotonic discoverability implies $A^*\equiv \left\{i\in A:\overrightarrow{h^A_i}\not\in \mu_i^\varphi\left(h^A\right)\right\}$ is non-empty. This set contains each agent for which, for any preference profile consistent with the collective history $h^A$, the assignment given by $\varphi$ is different from her last chosen object type. Each one of these agents in $A^*$ are given a menu containing the set of object types that are still feasible for them, under $\varphi$, for the preference profiles consistent with the collective history $h^A$.
        \end{itemize}
\end{itemize}

Therefore, for a given rule $\varphi$, the canonical PAO mechanism gives a simple recipe for how the menus should be constructed. For some rules the resulting mechanism is very intuitive.\footnote{The Iterative Deferred Acceptance Mechanism in which each student submits at most one college at a time \citep{Bo2016-id}, for example, is the canonical PAO mechanism for the DA rule.} For some others, such as the simple serial dictatorship (SD), it might be less so, and alternatively formulated PAO implementations are ``simpler'' to understand (see section \ref{sec:Experiments}).

Notice that the canonical PAO mechanism implies some additional properties given in the remark below.

\begin{rem}
\label{rem:PropertiesPAOSequentialization}
For every collective history $h^A$ that results from using a canonical PAO mechanism, an object $o$ is in a menu given to an agent $a$ after $h^A$ if and only if there is a continuation history of $h^A$ for which $a$ is assigned object $o$.
\end{rem}

In other words, since menus are always subsets of previous menus and agents are assigned to the last object they chose, it must be the case that every object type that is feasible after that point must be in any menu given to agents. Moreover, canonical PAO mechanisms only present object types that are feasible---conditional on the collective history---to the agent facing it.

Another fact worth noting is that when a rule is individually rational, every menu given by the canonical PAO mechanism contains the element $\emptyset$, and whenever an agent chooses it, her assignment is determined to be the null object---that is, she is left without any object. Individual rationality, in fact, is a property intrinsically connected with the idea of sequentializing rules using PAO mechanisms. If a rule is not individually rational, this implies that agents might have to make choices from menus that include only unacceptable objects. Since their outcomes in general depend on these choices, we would face a situation in which agents have to make ``correct'' outcome-relevant choices among unacceptable objects, an arguably undesirable property for this family of mechanisms. 

Monotonic discoverability therefore gives a full characterization of the rules for which it is possible to use a PAO mechanism to ``sequentialize'' the process of obtaining information on the participants' preferences, and Proposition \ref{prop:generalizedDAMonDisc} gives us a family of mechanisms that satisfy that condition. This does not, however, guarantee that it will be in the participants' own interest to truthfully reveal their preferences. When that is not the case, then the outcomes produced by these PAO mechanisms may differ substantially, with respect to the agents' true preferences, with that determined by the rule being used. In other words, we also need to consider the implementation problem when using PAO mechanisms. 

\section{Implementation in Pick-an-Object Mechanisms}

In this section we will consider the game that is induced by PAO mechanisms, and conditions on the rules that guarantee that it is in the participants' own interest to reveal their true preferences using straightforward strategies. The equilibrium concept that we will use is that of perfect ex-post equilibrium \citep{bergemann_robust_2005,Ausubel2004-bc}. Implementation in perfect ex-post equilibrium is considered a robust implementation notion, since an agent' optimal behavior does not depend on unrealistic assumptions about their knowledge of other players. Straightforward strategies being a perfect ex-post equilibrium implies that agents choosing from menus according to their true preferences is a mutual best-response regardless of previous behaviors, and for \emph{every preference profile}. 

Let $a$ be any agent, $\mathbb{S}$ be a PAO mechanism, and $H^A_{\mathbb{S}}$ be the set of collective histories of $\mathbb{S}$.  Moreover, let $H^A_{\mathbb{S}}(a)\equiv \bigcup_{h^A\in H^A_{\mathbb{S}}} h^A_a$. That is, $H^A_{\mathbb{S}}(a)$ is the set of all choice histories that agent $a$ might experience while interacting with that mechanism. We define an agent $a$'s \textbf{strategy for $\mathbb{S}$} to be a function $\sigma_a:H^A_{\mathbb{S}}(a)\cup \{\emptyset\}\times 2^{O} \to O$, where for every $h\in H^A_{\mathbb{S}}(a)\cup \{\emptyset\}$, $\sigma_a(h,\Omega)=o$ implies that $h \oplus ( \Omega,o)\in H^A_{\mathbb{S}}(a)$. In other words, for each agent, strategies allow agents to condition their choices on past choice histories and menus that they faced.

This representation of strategies is intuitive and simple: agents can condition their choices based on their past experience while interacting with the mechanism and the menu they face. It also reflects an environment in which agents are not given any information about the other agents' choices except for the ones that the agents can infer from the mechanics of the mechanism being used and their experience using it.\footnote{For instance, when the PAO mechanism being used sequentializes a serial dictatorship, an agent can infer the choice made by those higher than her in the priority order by looking at the missing items in the menu she receives. Her strategy could, therefore, condition her choice on these agents' actions in this case.}

A \textbf{type-strategy for $\mathbb{S}$} is a function $\Sigma$ that maps each preference ranking in $\mathcal{P}$ to a strategy for $\mathbb{S}$. Define the \textbf{outcome starting from $h^A$ under $\sigma$}, $\left.\mathcal{O}\right|_{h^A}\left(\sigma\right)$, as the assignment that results from following the PAO mechanism $\mathbb{S}$ when agents follow the strategies for $\mathbb{S}$ in $\sigma$ starting from the collective history $h^A\in H^A_{\mathbb{S}}\cup\{h^{A-\emptyset}\}$. Define moreover $\left.\mathcal{O}_a\right|_{h^A}\left(\sigma\right)$ as agent $a$'s outcome under that profile.

\begin{defn}
A type-strategy profile $\left(\Sigma_a\right)_{a\in A}$ for $\mathbb{S}$ is a \textbf{perfect ex-post equilibrium} if for every $a\in A$, $h^A\in H^A_{\mathbb{S}}\cup\{h^{A-\emptyset}\}$, strategy $\sigma'_a$ for agent $a$ and $P\in \mathcal{P}^{|A|}$:

\[\left.\mathcal{O}_a\right|_{h^A}\left(\Sigma_a(P_a),\left(\Sigma_{a'}(P_{a'})\right)_{a'\in A\backslash \{a\}}\right)\  R_a\  
\left.\mathcal{O}_a\right|_{h^A}\left(\sigma'_a,\left(\Sigma_{a'}(P_{a'})\right)_{a'\in A\backslash \{a\}}\right)\]
\end{defn}

We say that an allocation rule $\varphi$ is \textbf{pick-an-object implementable in perfect ex-post equilibrium} if there exists a PAO mechanism $\mathbb{S}$ that sequentializes $\varphi$, in which straightforward strategies constitute a perfect ex-post equilibrium.

\begin{thm}\label{thm:PAOTruthfullIFSP+MD}
Let $\varphi$ be an individually rational rule. Then $\varphi$ is pick-an-object implementable in perfect ex-post equilibrium if and only if $\varphi$ is strategy-proof and satisfies monotonic discoverability. Moreover, straightforward strategies constitute a perfect ex-post equilibrium profile for every pick-an-object mechanism that sequentializes a rule that is pick-an-object implementable in perfect ex-post equilibrium.
\end{thm}

Theorem \ref{thm:PAOTruthfullIFSP+MD} says, therefore, that strategy-proofness and monotonic discoverability fully characterize the individually rational rules that can be pick-an-object implemented in perfect ex-post equilibrium. Moreover, it shows that \textit{every} PAO mechanism that sequentializes each of these rules implements them in perfect ex-post straightforward equilibria. The implementation result, therefore, does not rely on some particular PAO mechanism. The canonical PAO mechanism introduced in section \ref{subsec:canonicalPAO}, for example, can be used with that objective for any of these rules.\footnote{In a previous version of this paper, this result was shown under a different solution concept---(Robust) Ordinal Perfect Bayesian Equilibrium. In its robust form, that concept is equivalent to perfect ex-post equilibrium. \cite{Mackenzie2020-gv} presents simultaneous and independent work showing that this result holds for Menu Mechanisms, a generalization of PAO mechanisms that does not necessarily assign agents to the objects they choose.} 

Notice that the definition of perfect ex-post equilibrium points to unique actions at every decision node in the game induced by the mechanism. Therefore, in an equilibrium with straightforward strategies, players simply follow a simple and myopic truthful rule of thumb and expect that to remain the case for any deviation that they might make.

Here, it is important to note that monotonic discoverability and strategy-proofness are independent. Firstly, it is easy to see that the Boston mechanism is a generalized DA procedure (and therefore satisfies monotonic discoverability) while famously not being strategy-proof. Below, we show an example of a rule that is strategy-proof but does not satisfy monotonic discoverability:

\begin{example}[\emph{A rule that is strategy-proof but not monotonic discoverable}]
\label{example:SPbutNotMonotDisc}
 Let $O=\{o_1,o_2,o_3,o_4\}$ and $\varphi$ be defined for two agents $a_1,a_2$. Agent $a_1$ is assigned her most-preferred object, and $a_2$ is assigned her most-preferred object among the two least-preferred objects by $a_1$.
\end{example}

\subsection{Implementation via unique straightforward equilibrium}
\label{sec:uniqueEquilibrium}

While we could make the case that Theorem \ref{thm:PAOTruthfullIFSP+MD} already provides a good foundation for using PAO mechanisms when implementing rules, we will show that some rules can be implemented via a \textbf{unique} perfect ex-post equilibrium. 

The first result involves a condition on the allocation rule:

\begin{defn}
\label{defn:RecourceMonotonicity}
Let $a^*$ be an agent and $P$ be a preference profile. Let $P'_{a^*}$ be such that  $\varphi_{a^*}(P'_{a^*},P_{-a^*})=\emptyset$. If the rule $\varphi$ satisfies \textbf{resource monotonicity}, then for every $a\in A\backslash\{a^*\}$, $\varphi_{a}(P'_{a^*},P_{-a^*})\ R_a\ \varphi_{a}(P)$. 
\end{defn}

Resource monotonicity, therefore, says that if an agent changes her preferences and as a result either remains or becomes unmatched, other agents cannot be made worse off. An example of a class of rules that are resource monotonic (and strategy-proof) are generalizations of the serial dictatorship, which can be defined by a procedure in which agents take turns choosing their most preferred allocation among the set of objects not yet assigned, but (i) depending on the agent and the object, choosing an object might also remove other objects from the set of available objects,\footnote{Consider, for example, a problem in which agents are assigned software packages, but whenever an agent with a hearing disability chooses a multimedia software, a transcription software is automatically bundled in---and therefore removed from the list of remaining softwares for the other agents.} and (ii) if an agent chooses $\emptyset$, then the set of objects that remain available does not change.

This condition allows us to obtain the following result:

\begin{prop}\label{prop:PAOUniqueRM}
Let $\varphi$ be an individually rational rule that is resource monotonic, strategy-proof and satisfies monotonic discoverability. Moreover, let $\mathbb{S}$ be the canonical PAO mechanism that sequentializes $\varphi$. Agents following straightforward strategies constitute the \textbf{unique} perfect ex-post equilibrium of the game induced by $\mathbb{S}$.
\end{prop}


The role that resource monotonicity has in Proposition \ref{prop:PAOUniqueRM} is guaranteeing that if, starting at some period, other agents follow strategies that leave them unmatched, this cannot make an object unavailable for someone else. This implies that at any step of a PAO mechanism, for each agent, there is a preference profile for the other players for which any choice she makes is final.

Next, we show that when implementing the Gale-Shapley DA rule via the canonical PAO mechanism, straightforward strategies constitute the unique ex-post equilibrium.\footnote{Note that the Gale-Shapley DA rule is not resource monotonic.}

\begin{prop}\label{prop:DAUnique}
Let $\varphi$ be the Gale-Shapley DA rule for an arbitrary list of capacities and priorities for the colleges. Moreover, let $\mathbb{S}$ be the canonical PAO mechanism that sequentializes $\varphi$. Agents following straightforward strategies constitute the \textbf{unique} perfect ex-post equilibrium of the game induced by $\mathbb{S}$.
\end{prop}

While uniqueness of equilibrium outcomes---i.e., full implementation---is a common objective that is achievable in some settings and for some rules, uniqueness of equilibrium strategies is less often so. It is worth comparing Proposition \ref{prop:DAUnique} with related results in the literature. \cite{Kawase2019-gf} show that, in a perfect information PAO implementation of Gale-Shapley DA, the student-optimal stable matching is the unique subgame-perfect Nash equilibrium outcome. \cite{Bo2016-id} also show that, under an additional behavioral assumption, straightforward strategies constitute a unique ex-post equilibrium for the Iterative Deferred Acceptance mechanism (a PAO implementation of Gale-Shapley DA). \cite{Mackenzie2020-gv} show that the Gale-Shapley DA rule can be ex-post ``everywhere-dominant'' (but not uniquely) implemented via a menu mechanism (which is also, in that case, a pick-an-object mechanism) in which only previous choices are removed from the participants' menus.

\subsection{Relation with Obviously Strategy-proof Mechanisms}\label{subsec:relationWithOSP}

When introducing the concept of obvious strategy-proofness (OSP), \cite{Li2017-oa} defined an OSP mechanism as one that has an equilibrium in obviously dominant strategies. A rule $\varphi$ is OSP-implementable if there is a game and an obviously dominant strategy for each type of player in that game, such that the outcome produced by this strategy for each type profile is what is determined by $\varphi$. \cite{Pycia2019-vl} showed that, in an object allocation environment such as the one that we use, every OSP-implementable rule could be implemented via a \textit{millipede game} with a \textit{greedy strategy}. Millipede games are sequential games where, in each period, some agent can either ``pass'' or ``clinch'' one of potentially multiple options in a menu which corresponds to private allocations that they can guarantee (in our setup, therefore, that agent can clinch an object and leave with it). A greedy strategy consists of an agent choosing to ``pass'' as long as her most-preferred object can still be clinched from a menu in some continuation history, and clinching it whenever it is in a given menu.\footnote{While each option in a menu given to a player in a millipede game corresponds to a private allocation for that player, there might be multiple options in that menu associated with the same allocation. While every option associated with an allocation results in the same outcome for the player making that choice, different choices among these might result in different outcomes for other players. Our results will consider only situations in which that is not the case: different items in a menu correspond to different private allocations for the agent making the choice. Therefore, we will ignore the possibility that they do not.}

The first thing to note is that there is a direct relation between a greedy strategy in a millipede game and a straightforward strategy in a PAO mechanism. In a millipede game, every time an agent interacts with the mechanism she is given a menu of objects that she can pick and keep for good and, potentially, also one option to ``pass.'' That is, as in PAO mechanisms, in millipede games an agent is always matched to the last object that she chose. However, in a millipede game, an agent can only choose an object once, and it remains her final allocation. In contrast, in a PAO mechanism, a chosen object can be interpreted as her ``tentative allocation,'' and thus she can potentially choose multiple times. The authors show that the greedy strategy is obviously dominant in that game. 

As long as an agent can infer the set of clinchable objects in continuation histories, following the greedy strategy is very simple: the only object she chooses from a menu is her most preferred in the set of feasible objects. Otherwise, she simply passes.\footnote{\cite{Pycia2019-vl}  question, however, the simplicity of inferring the set of clinchable objects in continuation histories, as it requires foresight from agents. They then introduce the concept of strong obvious strategy-proofness. In a strong OSP mechanism, the agents face a menu to choose from only once, and there is no passing option. The obviously dominant strategy in these games thus simply requires choosing the best object from the single menu that is offered.}

In a canonical PAO mechanism, on the other hand, menus always contain the entire set of feasible objects. Therefore, an agent following a straightforward strategy chooses the most-preferred object among those that are feasible, and will only have to make another choice once, and then only if that object is no longer feasible. Notice that, as opposed to millipede games, in these PAO mechanisms agents do not have to infer the set of feasible objects: all feasible objects are offered in the menus. Our next result relies on the common property of non-bossiness, defined below.

\begin{defn}
\label{defn:non-bossiness}
A rule $\varphi$ is \textbf{non-bossy} if, for every $a\in A$, if $\varphi_a\left(P_a,P_{-a}\right)=\varphi_a\left(P_a',P_{-a}\right)$, then $\varphi\left(P_a,P_{-a}\right)=\varphi\left(P_a',P_{-a}\right)$.
\end{defn}

If a rule is non-bossy and OSP implementable, we obtain the following relation with PAO mechanisms:\footnote{In their Supplementary Appendix, \cite{pycia2017incentive} provide a definition of their trading cycle mechanism that is a generalized DA procedure. Since trading cycle mechanisms characterize strategy-proof and non-bossy mechanisms, Theorem \ref{thm:PAOTruthfullIFSP+MD} together with their characterization implies that strategy-proof and non-bossy mechanisms are implementable in perfect ex-post equilibrium via PAO mechanisms. Proposition \ref{prop:OSPImplementableViaPAO} strengthens that result for implementation in weakly dominant strategies when the rule is also OSP implementable.}

\begin{prop}\label{prop:OSPImplementableViaPAO}
Every non-bossy OSP implementable rule is PAO implementable in weakly dominant strategies.
\end{prop}

One could make the argument that these equilibrium strategies in OSP mechanisms are ``simple and intuitive'': choose the most preferred object available, whenever that is possible. Given this, one possible explanation for the high rate of behavior in line with its prediction that OSP mechanisms have exhibited in previous laboratory experiments are the mechanics of \textit{picking your favorite object} in the equilibrium strategies, as opposed to the more general properties of obviously dominant strategies, in which less counterfactual reasoning is required to identify them.

If these simple mechanics explain at least in part this ``better'' behavior, then straightforward strategies in PAO mechanisms---which share these simple mechanics---could also lead agents to behave more in line with its theoretical prediction.

\section{Experiments\label{sec:Experiments}}
In this section, we present experiments designed to test the performance of PAO mechanisms for implementing different allocation rules, when compared to the traditional direct revelation mechanisms and OSP mechanisms, when possible. We chose two rules: top trading cycles (TTC) and serial dictatorship (SD). Ideally, we would have liked to compare implementing multiple allocation rules via PAO versus multiple alternative mechanisms. Our choice of using SD and TTC is driven mainly by the importance of both rules for practical applications. Another contender could be DA, but the comparison of the PAO implementation of DA versus its direct counterpart was made in two recent papers \cite{Bo2020-jm} and \cite{Klijn2019-bq}. Both led to a similar conclusion, reporting the better performance of the PAO mechanism, with respect to the truth-telling rates and stability of the allocations, when compared to the direct DA mechanism. 
Additionally, our choice of allocation rules was driven by the possibility of the OSP implementation of both rules (in the case of acyclic priorities for TTC). This allows us to compare PAO versus OSP implementations in a restricted setup.

\subsection{Mechanisms}
\label{subsec:ExperimentsMechanisms}

In this subsection, we describe the six mechanisms used in the experiments. For each allocation rule, namely TTC and SD, we use three different mechanisms: Direct,\footnote{We capitalize "Direct" when referring to the treatment name.} PAO, and OSP. For the PAO mechanisms, we used the canonical mechanism induced by these rules (see section \ref{subsec:canonicalPAO}). The mechanisms correspond to treatments in the experiments. Note that we describe the mechanisms in the same way they were described to the participants. We omit the description of actions when some or all objects are not acceptable for simplicity and to prevent confusion. The reason is that all objects lead to a positive payoff for the subjects in the experiments. Moreover, in the experiments, subjects had to choose at least one object in all sequential mechanisms and list all objects in the rank-order lists in the direct mechanisms. Details of the information given to the participants in the experiment are shown in section \ref{subsec:ExperimentsDesign}.

\vspace{0.5cm}

\textbf{Direct TTC}

Every participant submits her rank-ordered list of objects to a central authority. The following steps are executed by the central authority, without any further participation from the subjects.

\begin{itemize}
\item Step 1.1 All participants point to the object at the top of their submitted rank-ordered lists. Each object points to the participant with the highest priority at that object.
\begin{itemize}
\item The mechanism looks for cycles. There is at least one cycle. All participants in the cycle are assigned the object they pointed to.
\end{itemize}

\item Step 1.2 The priorities of the objects are updated to account for assigned participants. Submitted rank-ordered lists are updated to account for assigned objects. Steps 1--1.2 are repeated until all objects are assigned.
\end{itemize}
Direct TTC is strategy-proof, and Pareto efficient \citep{Abdulkadiroglu2003-dx}.
\vspace{0.5cm}

\textbf{PAO TTC}

All participants are asked to pick one object from a menu with all objects.

\begin{itemize}
\item Step 1.1. All participants point to the object they picked. All objects point to the participant with the highest priority at that object.
\begin{itemize}
\item The mechanism looks for cycles. There is at least one cycle. All participants in the cycle receive the object they pointed to. 
\end{itemize}
\item Step 1.2 For the remaining participants, if their last picked object was already assigned, the participant is asked to choose a new object from a menu of remaining objects.  Steps 1--1.2 are repeated until all the objects are assigned.
\end{itemize}

In PAO-TTC, every participant following the straightforward strategy is a perfect ex-post equilibrium (see Theorem \ref{thm:PAOTruthfullIFSP+MD}). If all participants follow straightforward strategies, the allocation is Pareto efficient.

\vspace{0.5cm}

\textbf{OSP TTC}

Note that OSP TTC is only defined when the priorities of the objects are acyclic. We use the mechanism described in \citet{Troyan2019-ah}.
\begin{itemize}
\item Step 1.0. The mechanism first tentatively assigns each object to the participant with the highest priority at that object (i.e., the participant \textit{tentatively owns} the object). 
\item Step 1.1. One by one each participant who tentatively owns an object is asked whether one of the objects that she owns is her favorite object. There are two possible answers for each participant and object:
\begin{itemize}
    \item If "Yes," the corresponding participant receives the object. Go to Step 1.2.
    \item If "No," she is asked about the next object among the ones she tentatively owns. If the participant answers “No” to all tentatively owned objects, the algorithm moves to the next participant who owns at least one object.  If all participants who tentatively own at least one object say "No" to all owned objects, then each participant who tentatively owns at least one object is asked to pick one object among the objects she does not own.
    \begin{itemize}
        \item All participants point to the object they picked. Each object points to its owner. The mechanism looks for cycles. There is at least one cycle. All participants in the cycle receive the object they picked.
    \end{itemize}
  
\end{itemize}
\item Step 1.2 The priorities of the objects are updated to account for the participants who left.  Steps 1.1--1.2 are repeated until all objects are assigned.
\end{itemize}

In OSP TTC, the truthful strategy\footnote{More specifically, the truthful strategy consists of only saying "Yes" when the object is the most-preferred among the objects that are still available, and picking the most-preferred object when asked to pick from objects that she does not own.} is obviously dominant, leading to a Pareto efficient allocation.
\vspace{0.5cm}

Note that in mechanisms implementing the SD allocation rule, we use priority scores instead of using ordinal tables priorities. In the experiments, each subject knew her score, and knew that the higher the score, the higher her priority. Details are explained in the experimental design section.

\vspace{0.5cm}

\textbf{Direct SD}

Every participant submits her rank-ordered list of objects to a central authority. The following steps are executed by the central authority, without any further input from the subjects.

\begin{itemize}
\item Step 1. The participant with the highest priority score is assigned the top-ranked object on her list.
\item Step 2. The participant with the second-highest priority score is assigned the top object on her list, among the remaining objects.
....
\item Step N. The participant with the Nth-highest priority score is assigned whatever object remains.
\end{itemize}

Direct SD is Pareto efficient and strategy-proof.
\vspace{0.5cm}

\textbf{PAO SD}

\begin{itemize}
\item Step 1. All participants are asked to pick one object from a menu with all objects. 
\item Step 2. The participant with the highest priority score is assigned the object she picked. All other participants who chose this object are asked to pick a new object from a menu containing the remaining objects. 

\item Step 3. The participant with the second-highest priority score is assigned the last object she picked. All other participants with lower priority who chose this object are asked to pick a new object from the menu containing the remaining objects. 
....
\item Step N+1. The participant with the lowest priority score is assigned the last object she picked. 
\end{itemize}

In PAO-SD, every player using the straightforward strategy is a perfect ex-post equilibrium (see Theorem \ref{thm:PAOTruthfullIFSP+MD}). If all participants follow straightforward strategies, the allocation is Pareto efficient.
\vspace{0.5cm}

\textbf{OSP SD}
\begin{itemize}
\item Step 1. The participant with the highest priority is asked to pick an object and she is assigned the object she picked.

\item Step 2. The participant with the second-highest priority is asked to pick an object among the remaining ones, and she is assigned the object she picked.
....
\item Step N. The participant with the lowest priority is assigned to the last remaining object.
\end{itemize}

In OSP SD, straightforward strategies are strongly obviously dominant, leading to a Pareto efficient allocation \citep{Pycia2019-vl}.

\subsection{Experimental design}
\label{subsec:ExperimentsDesign}

In the experiment, there were eight objects and eight participants. In all treatments, participants received 22 euros if they were matched to their most-preferred object, 19 euros for their second most-preferred object, 16 euros for their third most-preferred object, and so on. Participants received 1 euro if they were matched to their least-preferred object.

Each session lasted for 21 rounds. At the end of the experiment, one round was randomly drawn to determine the participants' payoffs.

Each round represented a new market. The preferences used in each market were generated following the designed market idea of \citet{Chen2006-lv}. For each market, agents' ordinal preferences were generated from cardinal utilities by adding, for each object, a common and an idiosyncratic value. The common values for each object were drawn from the uniform distribution with the range $[0,40]$. The idiosyncratic values were drawn, for each object and each agent, from the uniform distribution with the range $[0,20]$. The agent's utility from being matched to the object was the sum of both components. The resulting utilities were transformed into ordinal preferences. The procedure above ensures some correlation between participants' preferences. All objects had priorities over participants. The priorities were independently drawn from uniform distributions, in each round.\footnote{For rounds with acyclic priorities, priorities were drawn from the set of acyclic priorities. A new draw of acyclic priorities was generated for each round.}

We ran three treatments between-subjects: Direct, PAO, and OSP. Comparing Direct and PAO is our main focus. We additionally ran OSP, which is only available in SD and for a restricted set of environments in TTC -- namely those with acyclic priorities -- to disentangle the effect of the sequential pick-an-object setup from the effect of the simpler strategic setup of OSP mechanisms. To have a comparison of the treatments under a variety of parameters, including different allocation rules, we ran three environments within-subjects: TTC with cyclic priorities, TTC with acyclic priorities, and SD. 

Note that, with respect to the TTC rule, general priorities (cyclic) is the most applicable environment, and therefore the main focus of our analysis for that rule. TTC with acyclic priorities is a simpler decision environment as the acyclic priorities ensure that, at any time of the mechanism, only two participants could be at the top of the priorities of all objects.\footnote{TTC cycles, therefore, could involve only one or two agents.} TTC with acyclic priorities is implementable through an OSP mechanism, thus allowing us to compare it to a PAO implementation of that rule. SD is arguably the simplest allocation rule and also allows implementation through the OSP mechanism, which also coincides with the setup of sequentially making choices from a menu. Unlike the general PAO mechanism, it does not require the simultaneous play of participants at the first step.

For the first 14 rounds of the experiment, participants were matched using the TTC rule. During the first seven of these rounds, they faced markets with cyclic priorities in the Direct and PAO treatments. Given the importance of this comparison, and to prevent learning effects from other environments implemented within-subjects, we always run cyclic TTC in the first seven rounds. The same markets were used for Direct and PAO treatments. Because there is no OSP mechanism for TTC with cyclic priorities, we also generated acyclic priorities for the first seven rounds of the OSP treatment. For rounds eight to 14,  TTC with acyclic priorities was used to match participants to objects. The same markets were used in all three between-subjects treatments.

In the first 14 rounds, participants observed the full priority tables of all objects. The provision of priority tables is necessary, as it allows participants to see when they are acyclic in the OSP TTC. Participants knew only their own preferences, however, and not the preferences of other participants. They were aware that other participants might have the same or different preferences. We chose this informational environment to simplify the processing of market information, as providing complete information would lead to a longer decision time every round.

Finally, the SD allocation rule was used for the last seven rounds (rounds 15 to 21). Before round 15, subjects were informed about the switch of the mechanism. In the last seven rounds, the participants were assigned a priority score instead of observing the priority table of objects. They knew that the higher the score, the higher the priority. Participants were also informed that the priority scores would be drawn for each subject and each round independently from a uniform distribution with the range $[1,100]$ and each participant was informed of their own draw. This choice was made to ensure that even the participants with the lowest score had incentives to play a meaningful strategy. Otherwise, they would know that their choices were irrelevant for the allocation. As for preferences, just like in the first 14 rounds, participants only knew their own preferences and were informed that other participants might have different preferences. 

\begin{table}[]
    \centering
  \begin{tabular}{|c|c|c|c|}
\hline 
\textbf{Rounds} & \textbf{Direct} & \textbf{PAO} & \textbf{OSP}\tabularnewline
\hline
1-7 & TTC cyc & TTC cyc & TTC acyc\tabularnewline
\hline 
8-14 & TTC acyc & TTC acyc & TTC acyc\tabularnewline
\hline 
15-21 & SD & SD & SD\tabularnewline
\hline 
\end{tabular}
    \caption{Summary of treatments}
    \label{tab:treatm}
\end{table}

Table \ref{tab:treatm} presents the summary of the experimental design. Each cell of the table represents the mechanism used. Note that we explained to each subject two mechanisms to compare treatments under more than one allocation rule. It was feasible, as all SD mechanisms are quite simple and straightforward to explain; thus, we decided to run it within-subjects together with the TTC mechanisms. After each round, the participants learned the object they were matched to, but not the matches of the other participants.

The within-subjects design is driven by practical considerations and the goal of comparing treatments under various environments. We did not randomize the order of environments intentionally because we aimed at having high statistical power for the main comparison of the paper---PAO versus Direct under general priorities of TTC, given the number of participants. This environment was always run first and thus represents a clean comparison. The comparison between treatments under other environments might be affected by subjects' experience in the previous rounds. If subjects are more likely to play truthfully in the next round, given the truthful play in the previous round using the same kind of mechanism, it is possible that the design biases the results ``against'' the PAO treatment and ``in favor'' of the OSP treatment, since players in the OSP treatment play OSP TTC (with acyclic priorities) for 14 rounds, while in the PAO treatment, there were seven rounds with acyclic priorities, followed by seven rounds with cyclic priorities.

The experiment was run at the experimental economics lab at the Technical University of Berlin. We recruited student subjects from our pool with the help of ORSEE \citep{Greiner2015-ey}. The experiments were programmed in z-Tree \citep{Fischbacher2007-de}. Independent sessions were carried out for each of the three between-subjects treatments. Each session consisted of either 16 or 24 participants that were split into two or three groups of eight for the entire session. We use fixed groups in order to increase the number of independent observations and allow for maximum learning. As every round represents a new market and subjects play under incomplete information on the preferences of the other participants, subjects could not identify the strategies or identities of the players from previous rounds. Thus, we are not concerned about repeated games caveats.

In total, 15 sessions with 296 subjects were conducted. Thus, we have 96 subjects and 12 independent observations in Direct and PAO treatments, and 104 subjects and 13 independent observations for the OSP treatment. On average, the experiment lasted 110 minutes, and the average earnings per subject were 27 euros, including a show-up fee of 9 euros.

At the beginning of the experiment, printed instructions were given to the participants (see Appendix B). Participants were informed that the experiment was about the study of decision-making. The instructions were identical for all participants of each treatment, explaining the experimental setting in detail. First, the mechanism implementing the TTC rule was explained, with an example. The participants were told that this mechanism would be used to match them to objects for the first 14 rounds. Then, the mechanism implementing the SD rule was explained, also with an example, and participants were told that this mechanism would be used to match them to objects in the last seven rounds.  After round 14, participants were reminded about the mechanism switch. They were invited to re-read instructions for the second mechanism. Clarifying questions were answered in private. Note that the switch between cyclic and acyclic priorities does not switch the mechanism, and this change was not emphasized to subjects. They could, however, infer the difference from the priority tables.

\subsection{Experimental Results}

The significance level of all our results is 5\%, unless otherwise stated. For all tests, we use the p-values of the coefficient of the treatment dummy in regressions on the variable of interest. Standard errors are clustered at the level of matching groups, and the sample is restricted to the treatments that are of interest for the test.  We use the $>$ sign in the results between treatments to communicate "significantly higher," and the $=$ sign to communicate "the absence of a statistically significant" difference between treatments. 

\subsubsection{Truthful reporting}

This paper focuses on comparing the different proportions of equilibrium behavior induced by different mechanisms for the same environment.  We consider the proportions of subjects following a dominant strategy of truthful reporting in the direct mechanisms, a straightforward strategy that constitutes a perfect ex-post equilibrium in the PAO mechanisms, and the obviously dominant strategy in the OSP mechanisms. To simplify the language of distinguishing between these strategies, we use the concept of \textit{truthful strategy}. A participant follows a truthful strategy under a direct mechanism when she submits the truthful list of all eight objects. Note that there are typically non-truthful best responses as well, given the information available to the subjects taking part in the experiments. However, we argue that the submission of the full truthful list (truncations are not allowed by design) is the simplest strategy among them. Nevertheless, as robustness checks, we consider two alternative definitions of truthful strategies. The first one is \textit{truthful until the guaranteed object}, in which we deem as truthful any ranking that lists truthfully objects until the lowest-ranked object that the agent can receive regardless of other players' strategies.\footnote{For example, if an agent has the highest priority at her most-preferred object, this notion only requires the top object to be truthful in the ranking submitted. For SD, this notion requires the ranking to be entirely truthful, since every agent might be matched to their least-preferred object, depending on the priority ordering that is realized.} The second one is truthful until assigned object, that requires only the submission of a truthful list up to the assigned object in Direct. This is one of the most generous definitions for Direct, as it acts as if subjects correctly ``predict'' their assignment. 

In PAO mechanisms and OSP-SD, a participant following the truthful strategy is equivalent to her following the straightforward strategy. In OSP-TTC, participants following the truthful strategy must truthfully answer the ``yes-no'' questions about the most-preferred object among the ones for which the participant has the highest priority, and in case of all ``no'' answers, the participant must make a truthful choice of the favorite object among the other objects.

When considering truthful strategies in PAO mechanisms, there are subtle aspects that are worth highlighting. A strategy representing a preference different from its true one is perfectly distinguishable in Direct, but this is typically not true for a straightforward strategy in PAO. Since we only observe a limited number of choices and not the entire strategy, some non-straightforward strategies that coincide with straightforward ones in the observed choices will be deemed truthful, leading to a potential overestimation of the ``true'' truthful rate in PAO mechanisms. While we do not fully eliminate this possibility, the weaker notions of truthfulness used in Direct help mitigate the concern.

\noindent \textbf{Result 1 (Behavior in line with the truthful strategy):}
\begin{enumerate}
    \item Under the TTC rule with cyclic priorities, the comparison of average proportions of subjects behaving in line with the truthful strategy leads to the following results: PAO$>$Direct.\footnote{The result is robust to both alternative definitions of truthful strategy in the direct TTC.}
    \item Under the TTC rule with acyclic priorities, the comparison of average proportions of subjects behaving in line with the truthful strategy leads to the following results: OSP$>$PAO$>$Direct.\footnote{The result is robust to the notion of \textit{truthful until until the guaranteed object}. When considering \textit{truthful until the assigned object}, the difference between Direct and PAO is not significant.}
    \item Under the SD rule, the comparison of average proportions of subjects behaving in line with the truthful strategy leads to the following results: OSP$>$PAO$=$Direct.\footnote{The difference between Direct and PAO is no longer significant when considering \textit{truthful until the assigned object}. Note that under the informational conditions of the last seven rounds, there is no such thing as ``guaranteed objects.''}
\end{enumerate}
\medskip{}

\noindent \textbf{Support:}

\begin{centering}
\begin{figure}[ptb]
    \includegraphics[width=411pt]{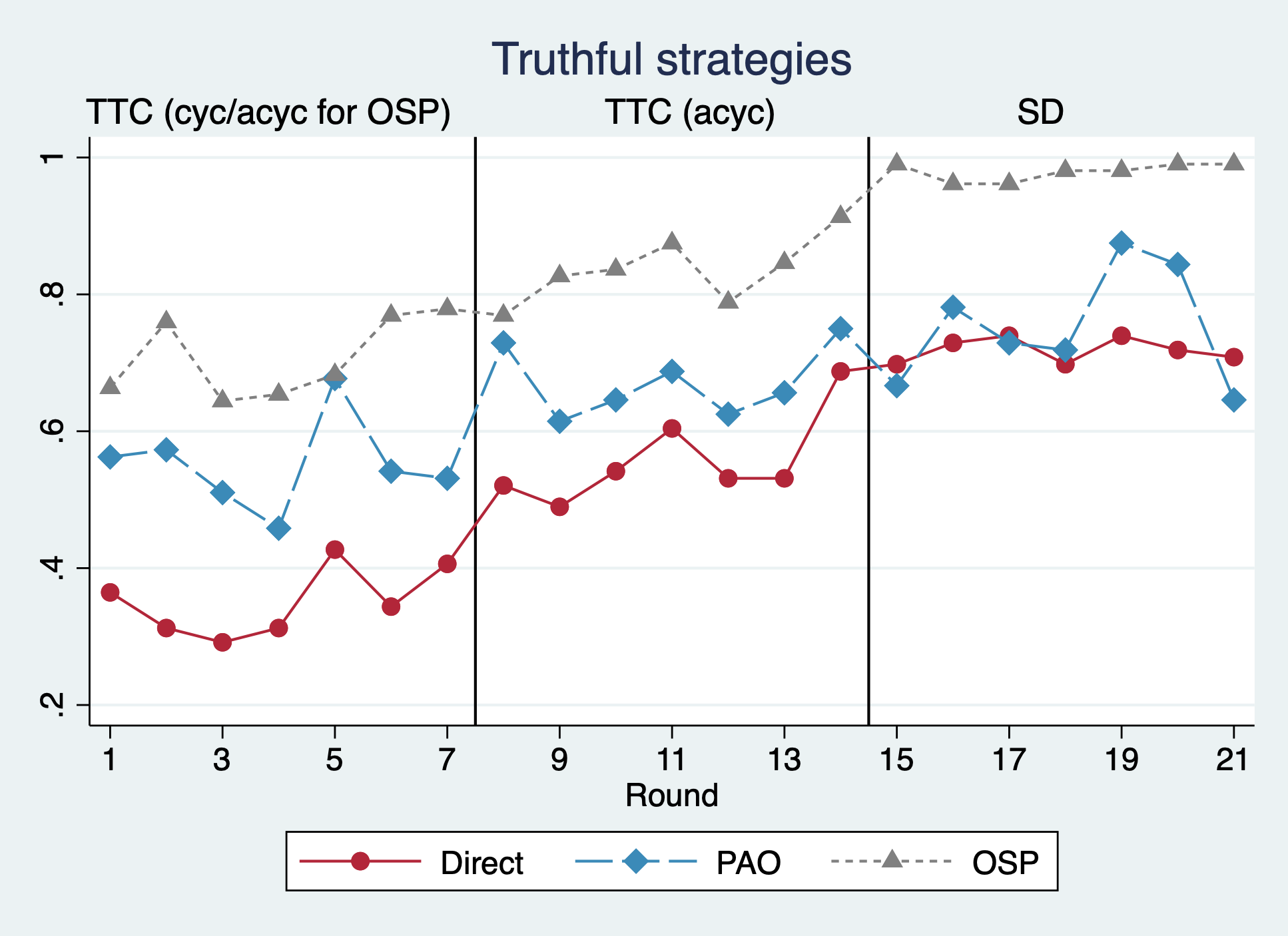}
    \caption{Truthful strategies by treatments and rounds}
    \label{truthful}
\end{figure}
\end{centering}

Figure \ref{truthful} presents the proportions of truthful strategies played by participants by treatments and rounds. 

\begin{table}[]
    \centering
    \begin{tabular}{|c|c|c|c|c|c|c|}
\hline 
\hline
 & Direct & PAO & OSP & Direct=PAO & Direct=OSP & PAO=OSP\tabularnewline
   &  &  &  & p-value & p-value & p-value\tabularnewline
\hline
\hline
\multicolumn{7}{c}{Panel A. Truthful}\tabularnewline
\hline
\hline
TTC cyclic & 35\% & 55\% & n/a & 0.00 & n/a & n/a\tabularnewline
\hline 
TTC acyclic & 56\% & 67\% & 77\% & 0.00 & 0.00 & 0.00\tabularnewline
\hline 
SD & 72\% & 75\% & 98\% & 0.47 & 0.00 & 0.00\tabularnewline
\hline
\hline
\multicolumn{7}{c}{Panel B. Truthful until guaranteed object}\tabularnewline
\hline
\hline

TTC cyclic & 36\% & 55\% & n/a & 0.00 & n/a & n/a\tabularnewline
\hline 
TTC acyclic & 59\% & 67\% & 77\% & 0.03 & 0.00 & 0.00\tabularnewline
\hline 
SD &          72\% & 75\% & 98\% & 0.47 & 0.00 & 0.00\tabularnewline
\hline 
\hline
\multicolumn{7}{c}{Panel C. Truthful until assigned object}\tabularnewline
\hline
\hline

TTC cyclic & 46\% & 55\% & n/a & 0.02 & n/a & n/a\tabularnewline
\hline 
TTC acyclic & 76\% & 67\% & 77\% & 0.53 & 0.00 & 0.00\tabularnewline
\hline 
SD & 76\% & 75\% & 98\% & 0.82 & 0.00 & 0.00\tabularnewline
\hline 
\end{tabular}
\begin{tablenotes}
\item \begin{flushleft}
\footnotesize{}{} \emph{Notes:} All the p-values are for the coefficient of the dummy for the corresponding treatment in the probit regression of the dummy for the truthful strategy on the treatment dummy, with the sample restricted to the treatments involved in the test. The standard errors of all regressions are clustered at the level of the matching groups. Thus, we have 24 clusters in regressions comparing Direct and PAO treatments, and 25 clusters in regressions that involve a comparison of the OSP treatment.

\end{flushleft}
\end{tablenotes}
    \caption{Proportions of truthful strategies}
    \label{tab:truth}
\end{table}

Panel A of Table \ref{tab:truth} presents the average proportion of truthful strategies and results of the test for treatment differences. Under the TTC rule with cyclic priorities, the average proportion of truthful strategies under direct TTC is 20 percentage points lower than under PAO TTC. The difference is significant. Panels B and C of Table \ref{truthful} presents the results when considering the proportions of truthful submissions in Direct until the guaranteed and assigned object, respectively.  The significance of the difference is robust to modifications of the definition of truthful strategy in the direct TTC. Note that, in the setup with cyclic priorities, no mechanism that implements the TTC rule is OSP \citep{Li2017-oa}. We observe that, despite the strategy-proofness of direct TTC, the proportion of truthful strategies is only 35\%, which is rather low. However, this rate is comparable to some other studies in the literature, that found similar rates of truthful reporting in TTC (for instance, 46\% in \citeauthor{Chen2006-lv}, \citeyear{Chen2006-lv} and 41\% in \citeauthor{Hakimov2018-qe}, \citeyear{Hakimov2018-qe}).\footnote{Two notable exceptions are \citet{Calsamiglia2010-ea} and \citet{Pais2008-gq}, who documented higher rates of truth-telling under TTC (62\% and 85\%, respectively). The high rate in \citet{Pais2008-gq} is likely driven by the fact that the rank-ordered list contained only three schools. Note, that in \citet{Hakimov2018-qe} in an environment with five schools, the rate is just 30\%.} Sequentialization of the rule through PAO leads to a significant increase in the proportion of truthful strategies. This finding is similar to the finding of \citet{Klijn2019-bq} and \citet{Bo2020-jm}, who show that the PAO implementation of the DA rule outperforms the direct mechanism. This is the main comparison of the experiment---in the absence of OSP alternatives, the PAO implementation can improve the proportions of equilibrium behavior relative to the implementation through the direct mechanism.

In the case of TTC with acyclic priorities, the OSP TTC outperforms both PAO and Direct mechanisms. The difference is significant for both the test considering only rounds 8 to 14, and the test considering rounds one to 14 in OSP TTC versus rounds eight to 14 in Direct and PAO TTC.\footnote{One can argue that the difference is driven by learning, as subjects in the OSP treatment had already played the OSP mechanism in the first eight rounds. We acknowledge the bias of our design in favor of the OSP treatment. Note, however, that the rate of truthful strategies in the OSP treatment in the first seven rounds is higher than the rate of truthful strategies under Direct and PAO in rounds 14-20, suggesting that the difference is unlikely to be entirely explained by learning.} As for the difference between Direct and PAO TTC, the proportion of truthful strategies is 9 percentage points higher under PAO TTC, and the difference is significant. Note, however, that the significance of the difference is not robust to redefining truthful strategies in direct TTC as the truthful ranking of objects in the rank-ordered lists until the \textit{assigned object}, as evident from Panel C of Table \ref{tab:truth}. The difference between Direct and PAO becomes smaller in TTC with acyclic priorities than in TTC with cyclic priorities. One can argue that learning is steeper under direct TTC. While this argument is hard to reject formally, suggestive evidence goes against this argument. More specifically, the coefficient for the variable ``round'' in the probit regression of the dummy for the truthful strategy is not significant in either Direct or PAO TTC. Also, there is a jump in the proportion of truthful strategies between rounds seven and eight under both treatments, suggesting that participants reacted to the switch of priorities from cyclic to acyclic.

For the SD rule, the use of the OSP mechanism results in almost universal (98\%) truthful behavior, as evident from Table \ref{tab:truth}. The rate of truthful strategies in OSP SD is significantly higher than in the Direct and PAO treatments. One can argue that the high rate of OSP is driven by the experience of subjects under OSP TTC, where the rate of truthful strategies was already higher. Again, the within-subjects feature of the design does not allow us to reject this argument formally. However, note that the decision environment under SD was quite different due to the different informational environment, and we also notified subjects about the switch of the mechanism before round 15. Also, the rate of truthful reporting under OSP SD is quite similar to \citet{hakimov2021costly} who run OSP SD between-subjects. There is no significant difference between Direct and PAO in all definitions of truthful strategies for Direct. SD is a simple allocation rule, and thus the rates of manipulations are already relatively low under the Direct mechanism. Note that under PAO SD, participants might engage in multiple decisions, especially when they have a low priority: every time the chosen object is taken by someone with a higher priority, the agent is asked to pick another one.

For a closer analysis of determinants of truthful strategies by treatments, see Appendix \ref{sec:App_truthful}. Next, we take a closer look at the determinants of higher truthful rates under OSP.

\smallskip{}

\noindent \textbf{Result 2 (The truthful strategy in OSP):} In OSP TTC, the rate of truthful behavior in the passing actions is much lower than in the clinching actions, and the difference is significant.\footnote{We follow the insights from \citet{Pycia2019-vl} and define ``passing actions" as OSP strategies, which involve saying ``no" to all objects which the agent tentatively owns under OSP TTC. Thus, the agent has to pass the decision to the next agent, hoping that a better object will appear in her choice set. ``Clinching actions" are OSP strategies that require only saying ``yes" to the most-preferred object in OSP TTC. Thus, whenever an agent is asked to act, she tentatively owns the most-preferred object.}

\medskip{}

\noindent \textbf{Support:} 

In line with the theory, obvious strategy-proofness leads to a higher rate of truthful strategies under both TTC and SD. Note that, in general, OSP-TTC is not strongly obviously strategy-proof, and the obviously dominant strategy might contain so-called ``passing" actions that require forward-looking from participants. In line with the definition of \citet{Pycia2019-vl}, we can categorize possible paths of the obviously dominant strategies of the participants into two groups: 

\begin{enumerate}
    \item \textbf{Clinching actions}---paths of the obviously dominant strategy that contain only clinching actions. If a participants was at the top of the priority of her favorite remaining object by the time the OSP TTC mechanism interviewed her (thus, she was at the top of the priority of at least one object). This strategy is strongly obviously strategy-proof, as it does not require passing to the other player, and thus does not require foresight from the participant.
    \item \textbf{Passing actions}---paths of the obviously dominant strategy that contain at least one passing action.  If a participant was not at the top of the priority of her favorite remaining object by the time the OSP TTC mechanism interviewed her (thus, she was at the top of the priority of at least one object). OSP requires to   ``pass" the interview turn to another participant. This path of the obviously dominant strategy does require foresight from the participant. In the context of OSP TTC, these strategies require saying ``no" to all objects and then picking the favorite object among those for which the participant does not have the top priority.
\end{enumerate}
 
 Table \ref{tab:clinchpass} presents the proportion of truthful strategies in OSP TTC, depending on the path of the OSP strategy. The rate of truthfulness in the passing actions is much lower than in the clinching actions. In fact, the rate of truthfulness in passing actions is not significantly different from truthful rates under direct TTC (p=0.94), and is significantly lower than under PAO ($p<0.01$), which are not obviously strategy-proof. In contrast, once the path of OSP contains only clinching actions like in strongly obviously strategy-proof strategies, the truthful rate is much higher than in other treatments and reaches 93\%.
 
\begin{table}[h]
    \centering
    \begin{tabular}{|c|c|c|}
\hline 
\multirow{2}{*}{} & \multicolumn{2}{c|}{OSP TTC acyclic}\tabularnewline
\cline{2-3} 
 & N & \% of truthful\tabularnewline
\hline 
Clinching actions& 836 & 93\%\tabularnewline
\hline 
Passing actions& 620 & 56\%\tabularnewline
\hline 
\end{tabular}
    \caption{Truthful strategies by clinching and passing in OSP TTC}
    \label{tab:clinchpass}
\end{table}

This result supports the concept of strong obvious strategy-proofness by \citet{Pycia2019-vl}. Indeed, when the market is such that the preferences and priorities of the object are strongly negatively correlated (the agents prefer objects that rank them the lowest), the obvious strategy-proofness of OSP TTC might not result in the high rates of optimal behavior, as most paths of the OSP strategy will contain passing actions.

Summing up the subsection on individual strategies, experimental results support using the PAO mechanisms in the complex environment, where an OSP mechanism is not available. Once the environment is simple enough to allow for the presence of the OSP mechanisms, they should be used. The benefit of the OSP mechanisms comes mostly through the presence of the paths in OSP that contain only clinching actions.

\subsubsection{Efficiency}
While our experiment focuses on the subjects' individual behavior, we look at the efficiency of the reached allocations in this subsection. On one hand, the matching game is often a zero-sum game, and a difference in efficiency is unlikely to appear.\footnote{Consider the theoretical argument of \citet{Liu2016-ts} for further details of this discussion.} Thus, as in many previous experiments, we do not expect to find large differences in efficiency between treatments. The differences might only appear due to Pareto-dominated allocations. 

\noindent \textbf{Result 3 (Efficiency):}
\begin{enumerate}
    \item Under the TTC rule with cyclic priorities: Direct$>$PAO.

    \item Under the TTC rule with acyclic priorities: Direct$>$PAO, OSP = PAO (round 8-14), OSP = Direct (round 8-14).

    \item Under the SD rule : Direct$>$OSP$=$PAO.

\end{enumerate}
\medskip{}
\noindent \textbf{Support:} 

Figure \ref{rank} shows the average rank of the assigned objects under the true preferences of participants by rounds. Thus, the higher the rank, the worse the assignment for participants. Table \ref{tab:ranks} shows the average ranks of assigned objects by treatments.  Under TTC with cyclic priorities (first seven rounds), the average rank of the assigned objects is significantly higher under Direct than under PAO. The difference is, on average, 0.29 of a rank. This is a large difference for matching markets experiments, as often a worse assignment of one participant leads to a better assignment of the other. Note that we do not present the results for the OSP treatment for the first seven rounds. Because the comparison would not be meaningful, as the participants played under different priorities, and thus the equilibrium allocations are different.

\begin{centering}
\begin{figure}[ptb]
    \includegraphics[width=411pt]{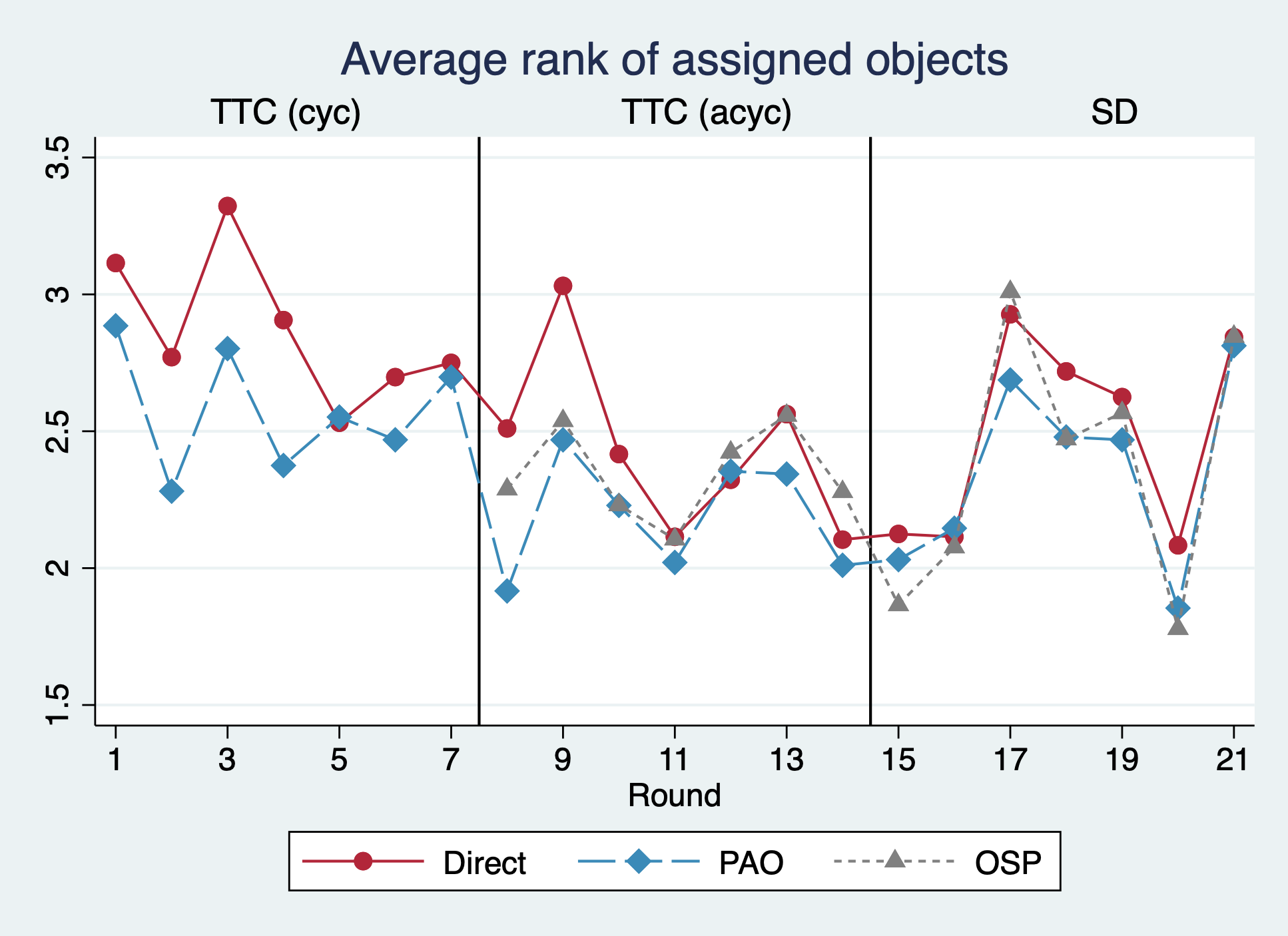}
    \caption{Average rank of assigned objects by treatments} 
        \label{rank}
\end{figure}

\end{centering}

\begin{table}[]
    \centering
    \begin{tabular}{|c|c|c|c|c|c|c|}
\hline 
 & Direct & PAO & OSP & Direct=PAO & Direct=OSP & PAO=OSP\tabularnewline
 &  &  &  & p-value & p-value & p-value\tabularnewline
\hline 
TTC cyclic & 2.87 & 2.58 & n/a & 0.00 & n/a & n/a\tabularnewline
\hline 
TTC acyclic & 2.43 & 2.19 & 2.34 & 0.00 & 0.33 & 0.08\tabularnewline
\hline 
SD & 2.49 & 2.35 & 2.37 & 0.02 & 0.02 & 0.66\tabularnewline
\hline 
\end{tabular}
\begin{tablenotes}
\footnotesize{}
\item \emph{Notes:} All the p-values are p-values for the coefficient of the
dummy for the corresponding treatment in the OLS regression of the rank of the assigned object in the true preferences of participants on the treatment dummy, with the sample restricted to the treatments involved in the test. The standard errors of all regressions are clustered
at the level of the matching groups. Thus, we have 24 clusters in regression comparing Direct and PAO treatments, and 25 clusters in regressions that involve comparison of OSP treatment.

\end{tablenotes}
    \caption{Average rank of assigned objects in the true preferences of the participants}
    \label{tab:ranks}
\end{table}

Under TTC with acyclic priorities, there is a small but statistically significant difference between PAO and Direct. Again, the average assignment is significantly better for participants under the PAO treatment. There is no significant difference between OSP and other treatments. Thus, a higher rate of truthful strategies in OSP does not result in a better average allocation of objects. 

Finally, under the SD rule, the average rank in the Direct mechanism is significantly higher than under OSP and PAO. Despite a similar rate of truthful strategies, the consequence of the deviations from truthful strategies is different between the Direct and PAO treatments. At the same time, despite the large difference in truthful rates between PAO and OSP, the average rank of assigned objects does not differ significantly.\footnote{For alternative definitions of efficiency and estimation of the costs of deviation from truthful strategies in each treatment see Appendix \ref{sec:App_eff}}

Summing up the subsection on efficiency, PAO outperforms Direct in all environments. Interestingly, PAO does not perform worse than OSP in both acyclic TTC and SD, despite lower rates of truthful strategies. This is because in PAO some deviations from truthful strategies are payoff-irrelevant, while under OSP they are more likely to be payoff-relevant under acyclic TTC and are always payoff-relevant under SD.

\section{Conclusion\label{sec:Conclusion}}

Recent empirical evidence raises concerns about the practical success of strategy-proof matching mechanisms inducing the truthful reporting of preferences in practical applications. Recent work by \citet{Li2017-oa} sheds a new light on the design of market mechanisms, emphasizing the importance of simpler and thus potentially more successful solutions for practice. However, the hope was not long-lasting, as many desirable allocation rules cannot be implemented via obvious strategy-proof mechanisms.

Our paper takes a different stand on potential solutions to the perceived complexity of direct mechanisms. We suggest using PAO mechanisms when a rule cannot be implemented via OSP mechanisms but belongs to an extensive family of rules, which include many commonly considered in practice and the literature. Similar to OSP mechanisms, PAO mechanisms can also implement those rules with an attractive and simple equilibrium strategy. 

Our experimental evidence, together with recent evidence by \cite{Bo2020-jm} and \cite{Klijn2019-bq}, show that improvement over direct mechanisms in allocations and the percentage of people following an equilibrium is possible for allocation rules for which OSP implementation is not available. These results might appear puzzling, and may invite further research on understanding the relative strength of different equilibrium concepts in predicting behavior, especially of inexperienced participants.

Finally, implementing rules via PAO mechanisms can be interpreted as a new type of interaction between a market designer and the participants.  Agents make requests to the designer, saying \textit{what they want},  as opposed to being asked questions about their ``types.'' The designer satisfies the request as long as the rule allows it. While it might be possible that the designer will return with bad news, this only happens when the designer cannot satisfy the request---and could, in principle, provide information justifying that decision.


\bibliographystyle{ecca}
\bibliography{schoolChoice}

\begin{thebibliography}{56}
\providecommand{\natexlab}[1]{#1}

\bibitem[{Abdulkadiro{\u g}lu and S{\"o}nmez(2003)}]{Abdulkadiroglu2003-dx}
\textsc{Abdulkadiro{\u g}lu, A.} and \textsc{S{\"o}nmez, T.} (2003). School
  choice: A mechanism design approach. \textit{American Economic Review},
  \textbf{93}~(3), 729--747.

\bibitem[{Akbarpour and Li(2020)}]{Akbarpour2019-lg}
\textsc{Akbarpour, M.} and \textsc{Li, S.} (2020). Credible auctions: A
  trilemma. \textit{Econometrica}, \textbf{88}~(2), 425--467.

\bibitem[{Artemov \textit{et~al.}(2017)Artemov, Che and He}]{Artemov2017-jx}
\textsc{Artemov, G.}, \textsc{Che, Y.-K.} and \textsc{He, Y.} (2017). Strategic
  'mistakes': Implications for market design research∗, {Working} {Paper}.

\bibitem[{Ashlagi and Gonczarowski(2018)}]{Ashlagi2018-zw}
\textsc{Ashlagi, I.} and \textsc{Gonczarowski, Y.~A.} (2018). Stable matching
  mechanisms are not obviously strategy-proof. \textit{Journal of Economic
  Theory}, \textbf{177}, 405--425.

\bibitem[{Ausubel(2004)}]{Ausubel2004-bc}
\textsc{Ausubel, L.~M.} (2004). An efficient ascending-bid auction for multiple
  objects. \textit{American Economic Review}, \textbf{108}~(2), 1452--1475.

\bibitem[{Ausubel(2006)}]{Ausubel2006-oc}
\textsc{---} (2006). An efficient dynamic auction for heterogeneous
  commodities. \textit{American Economic Review}, \textbf{96}~(3), 602--629.

\bibitem[{Aygün and Bó(2021)}]{Aygun2019-ko}
\textsc{Aygün, O.} and \textsc{Bó, I.} (2021). College admission with
  multidimensional privileges: The brazilian affirmative action case.
  \textit{American Economic Journal: Microeconomics}, \textbf{13}~(3), 1--28.

\bibitem[{Balinski and S{\"o}nmez(1999)}]{Balinski1999-uj}
\textsc{Balinski, M.} and \textsc{S{\"o}nmez, T.} (1999). A tale of two
  mechanisms: Student placement. \textit{Journal of Economic Theory},
  \textbf{84}~(1), 73--94.

\bibitem[{Bergemann and Morris(2005)}]{bergemann_robust_2005}
\textsc{Bergemann, D.} and \textsc{Morris, S.} (2005). Robust {Mechanism}
  {Design}. \textit{Econometrica}, \textbf{73}~(6), 1771--1813.

\bibitem[{B{\'o} and Hakimov(2020)}]{Bo2020-jm}
\textsc{B{\'o}, I.} and \textsc{Hakimov, R.} (2020). Iterative versus standard
  deferred acceptance: Experimental evidence. \textit{The Economic Journal},
  \textbf{130}~(626), 356--392.

\bibitem[{B{\'o} and Hakimov(2022)}]{Bo2016-id}
\textsc{---} and \textsc{Hakimov, R.} (2022). The iterative deferred acceptance
  mechanism. \textit{Games and Economic Behavior}, \textbf{135}, 411--433.

\bibitem[{B{\"o}rgers and Li(2019)}]{Borgers2019-dp}
\textsc{B{\"o}rgers, T.} and \textsc{Li, J.} (2019). Strategically simple
  mechanisms. \textit{Econometrica}, \textbf{87}~(6), 2003--2035.

\bibitem[{Breitmoser and Schweighofer-Kodritsch(2022)}]{Breitmoser2019-ur}
\textsc{Breitmoser, Y.} and \textsc{Schweighofer-Kodritsch, S.} (2022).
  Obviousness around the clock. \textit{Experimental Economics},
  \textbf{25}~(2), 483--513.

\bibitem[{Calsamiglia \textit{et~al.}(2010)Calsamiglia, Haeringer and
  Klijn}]{Calsamiglia2010-ea}
\textsc{Calsamiglia, C.}, \textsc{Haeringer, G.} and \textsc{Klijn, F.} (2010).
  Constrained school choice: An experimental study. \textit{American Economic
  Review}, \textbf{100}~(4), 1860--1874.

\bibitem[{Chen and Pereyra(2019)}]{chen2019self}
\textsc{Chen, L.} and \textsc{Pereyra, J.~S.} (2019). Self-selection in school
  choice. \textit{Games and Economic Behavior}, \textbf{117}, 59--81.

\bibitem[{Chen and S{\"o}nmez(2002)}]{Chen2002-tf}
\textsc{Chen, Y.} and \textsc{S{\"o}nmez, T.} (2002). Improving efficiency of
  on-campus housing: An experimental study. \textit{American Economic Review},
  \textbf{92}~(5), 1669--1686.

\bibitem[{Chen and S{\"o}nmez(2006)}]{Chen2006-lv}
\textsc{---} and \textsc{S{\"o}nmez, T.} (2006). School choice: an experimental
  study. \textit{Journal of Economic Theory}, \textbf{127}~(1), 202--231.

\bibitem[{Dur \textit{et~al.}(2018)Dur, Hammond and Morrill}]{Dur2018-cj}
\textsc{Dur, U.}, \textsc{Hammond, R.~G.} and \textsc{Morrill, T.} (2018).
  Identifying the harm of manipulable {School-Choice} mechanisms.
  \textit{American Economic Journal: Economic Policy}, \textbf{10}~(1),
  187--213.

\bibitem[{Echenique \textit{et~al.}(2016)Echenique, Wilson and
  Yariv}]{Echenique2016-cx}
\textsc{Echenique, F.}, \textsc{Wilson, A.~J.} and \textsc{Yariv, L.} (2016).
  Clearinghouses for two-sided matching: An experimental study.
  \textit{Quantitative Economics}, \textbf{7}~(2), 449--482.

\bibitem[{Fischbacher(2007)}]{Fischbacher2007-de}
\textsc{Fischbacher, U.} (2007). z-tree: Zurich toolbox for ready-made economic
  experiments. \textit{Experimental Economics}, \textbf{10}~(2), 171--178.

\bibitem[{Gale and Shapley(1962)}]{Gale1962-mn}
\textsc{Gale, D.} and \textsc{Shapley, L.~S.} (1962). College admissions and
  the stability of marriage. \textit{The American Mathematical Monthly},
  \textbf{69}~(1), 9--15.

\bibitem[{Gong and Liang(2016)}]{Gong2016-bn}
\textsc{Gong, B.} and \textsc{Liang, Y.} (2016). A dynamic college admission
  mechanism in inner mongolia: Theory and experiment, {Working} {Paper}.

\bibitem[{Greiner(2015)}]{Greiner2015-ey}
\textsc{Greiner, B.} (2015). Subject pool recruitment procedures: organizing
  experiments with {ORSEE}. \textit{Journal of the Economic Science
  Association}, \textbf{1}~(1), 114--125.

\bibitem[{Grenet \textit{et~al.}(2022)Grenet, He and
  K\"{u}bler}]{grenet2019decentralizing}
\textsc{Grenet, J.}, \textsc{He, Y.} and \textsc{K\"{u}bler, D.} (2022).
  Preference discovery in university admissions: The case for dynamic
  multioffer mechanisms. \textit{Journal of Political Economy},
  \textbf{130}~(6), 1427--1476.

\bibitem[{Guillen and Hakimov(2017)}]{Guillen2017-ft}
\textsc{Guillen, P.} and \textsc{Hakimov, R.} (2017). Not quite the best
  response: truth-telling, strategy-proof matching, and the manipulation of
  others. \textit{Experimental Economics}, \textbf{20}~(3), 670--686.

\bibitem[{Guillen and Hakimov(2018)}]{Guillen2018-cg}
\textsc{---} and \textsc{---} (2018). The effectiveness of top-down advice in
  strategy-proof mechanisms: A field experiment. \textit{European Economic
  Review}, \textbf{101}, 505--511.

\bibitem[{Haeringer and Hałaburda(2016)}]{haeringer_monotone_2016}
\textsc{Haeringer, G.} and \textsc{Hałaburda, H.} (2016). Monotone
  strategyproofness. \textit{Games and Economic Behavior}, \textbf{98}, 68--77.

\bibitem[{Haeringer and Iehle(2019)}]{Haeringer2019-rl}
\textsc{---} and \textsc{Iehle, V.} (2019). Gradual college admission,
  {Working} {Paper}.

\bibitem[{Hakimov \textit{et~al.}(2021{\natexlab{a}})Hakimov, Heller, Kübler,
  Kurino \textit{et~al.}}]{hakimov2019avoid}
\textsc{Hakimov, R.}, \textsc{Heller, C.}, \textsc{Kübler, D.},
  \textsc{Kurino, M.} \textit{et~al.} (2021{\natexlab{a}}). How to avoid black
  markets for appointments with online booking systems. \textit{American
  Economic Review}, \textbf{111}~(7), 2127--51.

\bibitem[{Hakimov and Kesten(2018)}]{Hakimov2018-qe}
\textsc{---} and \textsc{Kesten, O.} (2018). The equitable top trading cycles
  mechanism for school choice. \textit{International Economic Review},
  \textbf{59}~(4), 2219--2258.

\bibitem[{Hakimov and K{\"u}bler(2021)}]{hakimov2021experiments}
\textsc{---} and \textsc{K{\"u}bler, D.} (2021). Experiments on centralized
  school choice and college admissions: a survey. \textit{Experimental
  Economics}, \textbf{24}~(2), 434--488.

\bibitem[{Hakimov \textit{et~al.}(2021{\natexlab{b}})Hakimov, K{\"u}bler and
  Pan}]{hakimov2021costly}
\textsc{---}, \textsc{K{\"u}bler, D.} and \textsc{Pan, S.}
  (2021{\natexlab{b}}). Costly information acquisition in centralized matching
  markets. \textit{Working paper}.

\bibitem[{Hakimov and Raghavan(2021)}]{Hakimov2020-ig}
\textsc{---} and \textsc{Raghavan, M.} (2021). Improving transparency in school
  admissions: Theory and experiment, {Working} {Paper}.

\bibitem[{Hassidim \textit{et~al.}(2021)Hassidim, Romm and
  Shorrer}]{Hassidim2016-ls}
\textsc{Hassidim, A.}, \textsc{Romm, A.} and \textsc{Shorrer, R.~I.} (2021).
  The limits of incentives in economic matching procedures. \textit{Management
  Science}, \textbf{67}~(2), 951--963.

\bibitem[{Jones and Teytelboym(2017)}]{Jones2017-ae}
\textsc{Jones, W.} and \textsc{Teytelboym, A.} (2017). Matching systems for
  refugees. \textit{Journal on Migration and Human Security}, \textbf{5}~(3),
  667--681.

\bibitem[{Kagel \textit{et~al.}(1987)Kagel, Harstad and Levin}]{Kagel1987-ii}
\textsc{Kagel, J.~H.}, \textsc{Harstad, R.~M.} and \textsc{Levin, D.} (1987).
  Information impact and allocation rules in auctions with affiliated private
  values: A laboratory study. \textit{Econometrica}, \textbf{55}~(6),
  1275--1304.

\bibitem[{Kawase and Bando(2021)}]{Kawase2019-gf}
\textsc{Kawase, Y.} and \textsc{Bando, K.} (2021). Subgame perfect equilibria
  under the deferred acceptance algorithm. \textit{International Journal of
  Game Theory}, \textbf{50}~(2), 503--546.

\bibitem[{Klijn \textit{et~al.}(2019)Klijn, Pais and Vorsatz}]{Klijn2019-bq}
\textsc{Klijn, F.}, \textsc{Pais, J.} and \textsc{Vorsatz, M.} (2019). Static
  versus dynamic deferred acceptance in school choice: Theory and experiment.
  \textit{Games and Economic Behavior}, \textbf{113}, 147--163.

\bibitem[{Li(2017)}]{Li2017-oa}
\textsc{Li, S.} (2017). Obviously {Strategy-Proof} mechanisms. \textit{American
  Economic Review}, \textbf{107}~(11), 3257--3287.

\bibitem[{Liu and Pycia(2016)}]{Liu2016-ts}
\textsc{Liu, Q.} and \textsc{Pycia, M.} (2016). Ordinal efficiency, fairness,
  and incentives in large markets, {Working} {Paper}.

\bibitem[{Luflade(2018)}]{luflade2018value}
\textsc{Luflade, M.} (2018). The value of information in centralized school
  choice systems. \textit{Working paper}.

\bibitem[{Mackenzie and Zhou(2022)}]{Mackenzie2020-gv}
\textsc{Mackenzie, A.} and \textsc{Zhou, Y.} (2022). Menu mechanisms.
  \textit{Journal of Economic Theory}, \textbf{204}, 105511.

\bibitem[{Mandal and Roy(2022)}]{mandal2022obviously}
\textsc{Mandal, P.} and \textsc{Roy, S.} (2022). Obviously strategy-proof
  implementation of assignment rules: A new characterization.
  \textit{International Economic Review}, \textbf{63}~(1), 261--290.

\bibitem[{Milgrom and Weber(1982)}]{milgrom1982theory}
\textsc{Milgrom, P.~R.} and \textsc{Weber, R.~J.} (1982). A theory of auctions
  and competitive bidding. \textit{Econometrica: Journal of the Econometric
  Society}, pp. 1089--1122.

\bibitem[{Pais and Pint{\'e}r(2008)}]{Pais2008-gq}
\textsc{Pais, J.} and \textsc{Pint{\'e}r, {\'A}.} (2008). School choice and
  information: An experimental study on matching mechanisms. \textit{Games and
  Economic Behavior}, \textbf{64}~(1), 303--328.

\bibitem[{Pycia and Troyan(forthcoming)}]{Pycia2019-vl}
\textsc{Pycia, M.} and \textsc{Troyan, P.} (forthcoming). A theory of
  simplicity in games and mechanism design. \textit{Econometrica}.

\bibitem[{Pycia and {\"U}nver(2017)}]{pycia2017incentive}
\textsc{---} and \textsc{{\"U}nver, M.~U.} (2017). Incentive compatible
  allocation and exchange of discrete resources. \textit{Theoretical
  Economics}, \textbf{12}~(1), 287--329.

\bibitem[{Pycia and Yenmez(2021)}]{pycia2021matching}
\textsc{---} and \textsc{Yenmez, M.~B.} (2021). Matching with externalities.
  \textit{University of Zurich, Department of Economics, Working Paper},
  ~(392).

\bibitem[{Rees-Jones(2018)}]{Rees-Jones2018-uz}
\textsc{Rees-Jones, A.} (2018). Suboptimal behavior in strategy-proof
  mechanisms: Evidence from the residency match. \textit{Games and Economic
  Behavior}, \textbf{108}, 317--330.

\bibitem[{Roth and Peranson(1999)}]{Roth1999-id}
\textsc{Roth, A.~E.} and \textsc{Peranson, E.} (1999). The redesign of the
  matching market for american physicians: Some engineering aspects of economic
  design. \textit{American Economic Review}, \textbf{89}~(4), 748--780.

\bibitem[{Roth \textit{et~al.}(2004)Roth, S{\"o}nmez and
  {\"U}nver}]{Roth2004-pw}
\textsc{---}, \textsc{S{\"o}nmez, T.} and \textsc{{\"U}nver, M.~U.} (2004).
  Kidney exchange. \textit{Quarterly Journal of Economics}, \textbf{119}~(2),
  457--488.

\bibitem[{Schummer and Velez(2021)}]{Schummer2017-hh}
\textsc{Schummer, J.} and \textsc{Velez, R.~A.} (2021). Sequential preference
  revelation in incomplete information settings. \textit{American Economic
  Journal: Microeconomics}, \textbf{13}~(1), 116--47.

\bibitem[{Shapley and Scarf(1974)}]{Shapley1974-ir}
\textsc{Shapley, L.} and \textsc{Scarf, H.} (1974). On cores and
  indivisibility. \textit{Journal of Mathematical Economics}, \textbf{1}~(1),
  23--37.

\bibitem[{Shorrer and S{\'o}v{\'a}g{\'o}(2018)}]{Shorrer2018-uk}
\textsc{Shorrer, R.~I.} and \textsc{S{\'o}v{\'a}g{\'o}, S.} (2018). Obvious
  mistakes in a strategically simple college admissions environment: Causes and
  consequences, {Working} {Paper}.

\bibitem[{Troyan(2019)}]{Troyan2019-ah}
\textsc{Troyan, P.} (2019). Obviously strategy-proof implementation of top
  trading cycles. \textit{International Economic Review}, \textbf{60}~(3),
  1249--1261.

\bibitem[{Veski \textit{et~al.}(2017)Veski, Bir{\'o}, P{\"o}der and
  Lauri}]{Veski2017-ix}
\textsc{Veski, A.}, \textsc{Bir{\'o}, P.}, \textsc{P{\"o}der, K.} and
  \textsc{Lauri, T.} (2017). Efficiency and fair access in kindergarten
  allocation policy design. \textit{The Journal of Mechanism and Institution
  Design}, \textbf{2}~(1), 57--104.

\end{thebibliography}

\pagebreak{}

\appendix

\section{Proofs}

\subsection*{Proof of Theorem \ref{thm:PickAnObjectIFFTerminalAv}}

\begin{customthm}{\ref{thm:PickAnObjectIFFTerminalAv}}
There exists a pick-an-object mechanism that \textbf{sequentializes} an individually rational rule $\varphi$ if and only if $\varphi$ satisfies monotonic discoverability.
\end{customthm}

\begin{proof}

First we assume that $\varphi$ satisfies monotonic discoverability and show that the canonical pick-an-object function $\mathbb{S}$ is such that its mechanism \textbf{sequentializes} $\varphi$.

Define a pick-an-object function $\mathbb{S}$ such that for every $i$ and $h^A\in H^A$:

\[\mathbb{S}^i\left(h^A\right)=
\begin{cases}
    \emptyset & \text{ if } \left|\varphi\left(h^A\right)\right|=1 \text{ or } \overrightarrow{h^A_i}\in \mu_i^\varphi\left(h^A\right) \\
    \mu_i^\varphi\left(h^A\right) & \text{otherwise}
\end{cases}
\]

To show that $\mathbb{S}$ sequentializes $\varphi$, we will show that: (i) if the collective history $h^A$ is such that $\left|\varphi\left(h^A\right)\right|>1$, at least one agent must be given a non-empty menu, that is, there must be at least one $i$ such that $\mathbb{S}^i\left(h^A\right)\neq \emptyset$; (ii) When an empty menu is returned for all agents, their last choices must be the allocation that $\varphi$ determines what should be produced by any preference profile consistent with the collective history.

First, (i). Suppose not. Then $\left|\varphi\left(h^A\right)\right|>1$ and for all $i$, $\mathbb{S}^i\left(h^A\right)= \emptyset$. By the definition of $\mathbb{S}$, this implies that for all $i$, $\overrightarrow{h^A_i}\in \mu_i^\varphi\left(h^A\right)$. That is, the last choice of each agent is a feasible assignment after $h^A$. 

Let $\overrightarrow{\mu}$ be the allocation that matches each agent with her last choice in $h^A$, that is, for every $i$, $\overrightarrow{\mu}\left(a_i\right)=\overrightarrow{h^{A}_i}$, and $P^A$ be any preference profile consistent with $h^A$. Monotonic discoverability implies that either (a) for every preference profile $P'$ consistent with continuations of $h^A$, $\varphi\left(P'\right)=\overrightarrow{\mu}$, which is a contradiction with $\left|\varphi\left(h^A\right)\right|>1$, or (b) that there is at least one agent $a_{i*}$ such that for all $P\in\mathcal{L}\left(P^A,\overrightarrow{\mu}\right)$,  $\overrightarrow{h^{A}_{i*}}\neq \mu_{i*}\left(P\right)$, which is again a contradiction with $\overrightarrow{h^A_{i*}}\in \mu_{i*}^\varphi\left(h^A\right)$.

Now, to (ii). Since an empty menu is given to all agents, then by the definition of $\mathbb{S}$, either (a) $\left|\varphi\left(h^A\right)\right|=1$ or (b) for every $i$, $\overrightarrow{h^A_i}\in \mu_i^\varphi\left(h^A\right)$. Consider first (a). By definition of the notation, $\left|\varphi\left(h^A\right)\right|=1$ implies that for all preference profiles consistent with $h^A$, the rule $\varphi$ determines the same allocation. Suppose, however, that $\varphi\left(h^A\right)\neq\overrightarrow{\mu}$, and let $a_i$ be an agent for whom $\varphi_i\left(h^A\right)\neq \overrightarrow{\mu}\left(a_i\right)$. Clearly, $\varphi_i\left(h^A\right)$ cannot be any choice made before $\overrightarrow{\mu}\left(a_i\right)$ in $h^A$, since by design of the pick-an-object mechanism it is only rejected if it is not a feasible assignment anymore. So it can either be an object which was present in a menu previous to the one where $\overrightarrow{h^A_i}$ was chosen but not in some future menu, or in the menu given when $\overrightarrow{h^A_i}$ was chosen. 
The first option contradicts the way in which menus are constructed: menus contain all feasible assignments conditional on the collective history. So if at some point the object type was not feasible anymore, it cannot be that agent's assignment under $\varphi$. For the second, this implies that the allocation was determined by $\varphi$ to match agent $a_i$ to some object type $o^*$ that was not chosen.  Notice, however, that since $\varphi$ is individually rational, a collective history cannot point to a single allocation unless the agents' preferences consider these objects acceptable. That is, the collective history must include agents choosing, at some point, these objects from a menu that includes the option ``$\emptyset$''. Therefore, it cannot be the case that $\varphi$ is individually rational, $\left|\varphi\left(h^A\right)\right|=1$, and an $\varphi\left(h^A\right)$ matches some agent to an object that she did not choose from a menu. Now, case (b). Here, all the last choices of all agents are feasible in $\varphi\left(h^A\right)$. That is, there is no agent who will not be matched to their last choice for a preference profile that is consistent with $\varphi\left(h^A\right)$. But then monotonic discoverability implies that $\varphi\left(h^A\right)=\overrightarrow{\mu}$, which is what we wanted to show.

Next, we will show that if a rule $\varphi$ does not satisfy monotonic discoverability, then it cannot be sequentialized by some pick-an-object function $\mathbb{S}$. Suppose
not. Then there is a rule $\varphi^{*}$ that does not satisfy monotonic
discoverability and a pick-an-object mechanism $\mathbb{S}$ that sequentializes
it. Since $\varphi^{*}$ does not satisfy monotonic discoverability, there
exists a preference profile $P^{*}$ and an allocation $\mu^{*}$
such that (i*) $\varphi^{*}\left(P^{*}\right)\neq\mu^{*}$ and (ii*)
for each agent $a\in A$, there is at least one preference profile
$P^{*,a}\in\mathcal{L}\left(P^{*},\mu^{*}\right)$ where $\varphi_{a}^{*}\left(P^{*,a}\right)=\mu^{*}\left(a\right)$. If there is more than one such allocation $\mu^*$ for that given $P^*$, let $\mu^*$ be such that for every $\mu'\neq \mu^*$ satisfying (i*) and (ii*), there is at least one $a\in A$ such that $\mu^*(a) P^*_a \mu'(a)$.

The first thing to note next is that (ii*) implies that when all agents follow straightforward strategies with respect to $P^*$, there must be for each agent $a$ a period in which she chooses $\mu^*(a)$, and that this must happen before she chooses other objects that she is matched to by $\varphi^*$ for any profile in $\mathcal{L}\left(P^{*},\mu^{*}\right)$. That is, since when following $P^*$ ``up to $\mu^*$'' there is some continuation in which $a$ is matched to $\mu^*(a)$, then $a$ must first choose that object from a given menu.

The next observation is that, after the period in which $a$ chooses $\mu^*(a)$, the determination of whether an agent $a$ will be matched to $\mu^*(a)$ or some other object below $\mu^*(a)$ in her preference cannot depend on that agent's preferences among objects below $\mu^*(a)$ in her preference, since in order to obtain information about that part of agent $a$'s preferences \textit{requires rejecting $\mu^*(a)$ as a potential allocation for $a$}. This implies that whether $a$ will be matched to $\mu^*(a)$ or not depends on information about the \textit{other agents}' preferences. More than that, (ii*) specifically implies that the conclusion that $a$ will not be matched to $\mu^*(a)$ cannot be reached before some other agent $b$ makes choices \textit{after} having her choice of $\mu^*(b)$ rejected. This, therefore, has the following implications:

\begin{itemize}
    \item Every agent $a^*$ following the straightforward strategy with respect to $P^*$ will, in some period, choose her allocation under $\mu^*$,
    \item In any periods that follow, in which the other agents did not yet choose their allocation under $\mu^*$, agent $a^*$'s allocation may or not be determined to be $\mu^*\left(a^*\right)$, but will \textbf{not} have her choice rejected, since there is still some continuation in which $\mu^*\left(a^*\right)$ will be her allocation, and rejections are final.
\end{itemize}

The two implications above result in the following dynamic when agents follow straightforward strategies with respect to $P^*$: agents make choices over menus until, at some point, they choose their allocation under $\mu^*$. After a certain number of periods, therefore, we reach a point in which all agents' last choices are their allocations in $\mu^*$. By (i*), $\mu^*$ should not be the allocation, but by (ii*), for each agent, more information about the preference profile is necessary to point out the correct allocation to be produced. This requires rejecting at least one of the agents' choices, but by (ii) for every agent there is a continuation in which she is matched to her assignment under $\mu^*$. So no more information can be obtained when using any pick-an-object function, leading to a contradiction.

\end{proof}

\subsection*{Proof of Proposition \ref{prop:generalizedDAMonDisc}}

\begin{customprop}{\ref{prop:generalizedDAMonDisc}}
If $\varphi$ is described by a generalized DA procedure, then $\varphi$ satisfies monotonic discoverability.
\end{customprop}

\begin{proof}
In light of Theorem \ref{thm:PickAnObjectIFFTerminalAv}, it suffices to show that there is a pick-an-object mechanism that sequentializes the rule that is specified by the generalized DA procedure. Let $\Psi^*$ be the update function used to describe the rule $\varphi$. We construct the menu function $\mathbb{S}^*$ as follows:

\[\mathbb{S}^*\left(h^{A-\emptyset}\right)=\left(O,O,O,\ldots\right)\]

The value of $\mathbb{S}^*$ for other collective histories are determined, recursively, as follows. Let $h^A$ be a collective history for which the value of $\mathbb{S}^*\left(h^{A}\right)$ has already been determined as $\mathbb{S}^*\left(h^{A}\right)=\left(\phi_1,\phi_2,\ldots,\phi_n\right)$. 

For each choice profile $\left(o^1,o^2,\ldots,o^n\right)$, where for each $i\in A$, $o^i\in \phi_i$ if $\phi_i\neq\emptyset$ and $o^i=\diamondsuit$ otherwise,\footnote{Here the symbol $\diamondsuit$ is used as a placeholder for the agents who are not presented with a menu to choose from.} perform the following:

\begin{enumerate}
    \item Construct the assignments $\mu^1$ and $\mu^2$ as follows:
    \begin{itemize}
        \item For every $a_i\in A$:
    \begin{itemize}
        \item If $o^i=\diamondsuit$ and $h^A\neq h^{A-\emptyset}$, let $\mu^1(a_i)=\overrightarrow{h^A_i}$.
        \item Otherwise, let $\mu^1(a_i)=\emptyset$.
    \end{itemize}
    \item For every $a_i\in A$, let $\mu^2(a_i)=\emptyset$ if $o^i=\diamondsuit$, and $\mu^2(a_i)=o^i$ otherwise. 
    \end{itemize}
    \item For every $a_i\in A$, define the choice history $h_i$ to be:
    \begin{itemize}
        \item $h_i=h^A_i\oplus \left(\phi_i,o^i\right)$ if $o^i\neq \diamondsuit$,
        \item $h_i=h^A_i$ otherwise.
    \end{itemize}
    \item Let $\mu^3=\Psi^*\left(\mu^1,\mu^2\right)$.
    \begin{itemize}
        \item If for every $a\in A$ it is the case that $\mu^{t}(3)\in \left\{\mu^{1}(a),\mu^2(a)\right\}$, then $\mathbb{S}^*\left(h_1,h_2,\ldots,h_n\right)=\left(\emptyset,\emptyset,\ldots,\emptyset\right)$.
        \item Otherwise, $\mathbb{S}^*\left(h_1,h_2,\ldots,h_n\right)=\left(\phi'_1,\phi'_2,\ldots,\phi'_n\right)$, where for each $a_i\in A$:
        \begin{itemize}
            \item $\phi'_i=\emptyset$ if $\mu^3(a_i)\in \left\{\mu^1(a_i),\mu^2(a_i)\right\}$,
            \item $\phi'_i=O\backslash\bigcup_{(\Omega,\omega)\in h^A_i}\omega$ otherwise.
        \end{itemize}
    \end{itemize}
\end{enumerate}

Notice first that $\mathbb{S}^*$ is a menu function: initial menus are the entire set $O$, and for every collective history the menus given are precisely the last menu an agent was given with her last choice removed. Since $\mathbb{S}^*$ is defined recursively for each collective history that can be generated by choices from the menus that can be offered, every possible path of the pick-an-object mechanism is well defined.  Consider now the pick-an-object mechanism $\mathbb{S}^*$ and the agents in $A$ following straightforward strategies with respect to $P$. 

What follows is an exact reproduction of the steps of the generalized DA procedure under the preference profile $P$. Since the menus always include all elements of $O$ minus the agents' past choices, agents will make choices from the top until their last choice following their preference, as in the generalized DA. Agents who have their last choices rejected and are given new menus, for any collective history, are determined by the function $\Psi^*$, so that the way in which the sequences of choices from menus determine whether an agent is tentatively matched or not is given by that function. And finally, whenever the last choice of agents should be determined as the outcome, $\mathbb{S}^*$ returns a list of empty menus.

We can conclude, therefore, that the pick-an-object mechanism $\mathbb{S}^*$ sequentializes the rule $\varphi$.
\end{proof}

\subsection*{Proof of Theorem \ref{thm:PAOTruthfullIFSP+MD}}

The proof proceeds as follows: first, we assume that $\varphi$ is strategy-proof and satisfies monotonic discoverability, and show that straightforward strategies constitute a perfect ex-post equilibrium for every PAO mechanism that sequentializes $\varphi$. Next, we show that if $\varphi$ is not strategy-proof, no PAO mechanism that sequentializes $\varphi$ has a straightforward perfect ex-post equilibrium. Adding the fact that Theorem \ref{thm:PickAnObjectIFFTerminalAv} shows that monotonic discoverability is necessary for the existence of a PAO mechanism that sequentializes $\varphi$ and the fact that strategy-proofness does not imply monotonic discoverability, we have everything we need.

In order to show that straightforward strategies constitute a perfect ex-post equilibrium, we  need to show that even when past choices were not a result of following straightforward strategies, being truthful from that point forward is still a best-response for every preference an agent might have. Since the PAO mechanism sequentializes $\varphi$, this will involve comparisons of outcomes among non-truthful strategies. While at first sight these comparisons are not implied by any of the properties that we assume $\varphi$ has, the comparison we need can be done based on the strategy-proofness of $\varphi$. Some notation is needed first.

For any set $I\subseteq O$, let $I!$ denote the set of all permutations of the elements of $I$.\footnote{For example, if $I=\left\{o_1,o_2,o_3\right\}$, then $I!=\left\{(o_1,o_2,o_3),(o_1,o_3,o_2),(o_2,o_1,o_3),(o_2,o_3,o_1),(o_3,o_1,o_2),(o_3,o_2,o_1)\right\}$.}
For any tuple $\gamma$ of distinct elements of $O$ (for example,
$\gamma=\left(o_{4},o_{7},o_{2}\right)$), let:

\[\left.\mathbb{P}\right|_{\gamma}\equiv \bigcup_{\lambda\in \left(O\backslash\gamma\right)!} \gamma\oplus\lambda \]

That is, $\left.\mathbb{P}\right|_{\gamma}$ is the set of all preferences in which $\gamma$ are the most-preferred object types, ordered as in the tuple $\gamma$ itself. Let also $\left.P\right|_{I}$, where $I\subseteq O$, be
the preference $P$ restricted to $I$ (for example,
if $P=\left(o_{1},o_{4},\emptyset,o_{2},o_{3}\right)$, $\left.P\right|_{\left(o_{2},o_{4}\right)}=\left(o_{4},o_{2}\right)$). We can now define recursive dominance.

\begin{defn}
A rule $\varphi$ satisfies \textbf{recursive dominance} if for every agent $a$, preference $P_a$, preferences of other agents $P_{-a}$, set $I\subseteq O$, permutation $\gamma$ of $I$ and $P^*\in \left.\mathbb{P}\right|_{\gamma}$:

\[
\varphi_{a}\left( \gamma\oplus \left.P\right|_{O\backslash I} ,P_{-a}\right)\ R_{a}\ \varphi_{a}\left(P^*,P_{-a}\right)
\]
\end{defn}

The following corollary is derived from the fact that recursive dominance is implied by monotone strategy-proofness, a concept introduced by \cite{haeringer_monotone_2016}:

\begin{cor}{\citep{haeringer_monotone_2016}}
\label{cor:monotoneSP}
A rule $\varphi$ satisfies recursive dominance if and only if it is strategy-proof.
\end{cor}

The next lemma shows the outcomes that are produced when other agents follow straightforward strategies starting from a particular collective history.

\begin{lem}
\label{lem:outcomeAfterDeviation}
Let $\mathbb{S}^{\varphi}$ be a pick-an-object mechanism that sequentializes $\varphi$, and $h^A\in H^A_{\mathbb{S}}$ be any collective history of $\mathbb{S}^{\varphi}$. Let $P$ be any preference profile, $\sigma^{P}$ be the straightforward strategy profile with respect to $P$, $a^*$ be any agent,  and $\sigma'$ be any other strategy for $a^*$. Then there is a set of object types $I\subseteq O$, a tuple $\gamma$ with elements of $O$, a preference $P^*\in \mathcal{P}$, and a preference profile for agents other than $a^*$ $P'_{-a^*}$ such that:

\[\left.\mathcal{O}_{a^*}\right|_{h^A}\left(\sigma_{a^*}^{P},\sigma_{-a^*}^{P}\right)=\varphi_{a^*}\left(\gamma\oplus \left.P\right|_{O\backslash I}, P'_{-a^*}\right)\]

and

\[\left.\mathcal{O}_{a^*}\right|_{h^A}\left(\sigma',\sigma_{-a^*}^{P}\right)=\varphi_{a^*}\left(\gamma\oplus \left.P^*\right|_{O\backslash I}, P'_{-a^*}\right)\]
\end{lem}

\begin{proof}
Let $a$ be any agent, and take the element of $h^A$ containing her menus and choices:

\[h^A_a=\left( \left(\Omega_1,\omega_1\right), \left(\Omega_2,\omega_2\right), \ldots, \left(\Omega_{k},\omega_{k}\right) \right)\]

Consider the tuple $\gamma_a^h$ as follows, where the sets in double brackets are replaced by an arbitrary permutation of their elements:

\[
\gamma_a^h=\left(\omega_1\ ,\ \llbracket\Omega_1\backslash
(\{\omega_1\}\cup \Omega_2) \rrbracket\ ,\  \omega_2\ ,\ \llbracket\Omega_2\backslash(
\{\omega_2\}\cup \Omega_3 )\rrbracket\ , \cdots ,\ \omega_k
\right)
\]

Notice that if agent $a$ had preferences $P_a^h=\gamma_a^h\oplus \llbracket\Omega_k\backslash \{\omega_k\} \rrbracket$, she would have made these exact same choices when facing the menus that she faced in $h^A$. Consider now the continuation collective history $h^{A-P}$ that results from every agent following straightforward strategies with respect to $P$ after $h^A$. For any agent $a$, the element of $h^{A-P}$ containing her menus and choices has the following form:

\[h^{A-P}_a=h^A_a\oplus\left( \left(\Omega_{k+1},\omega_{k+1}\right), \left(\Omega_{k+2},\omega_{k+2}\right), \ldots, \left(\Omega_{\ell},\omega_{\ell}\right) \right)\]

Since we assume that $a$ followed the straightforward strategy, $\omega_{k+1}$ is the most-preferred element of $\Omega_{k+1}$ with respect to $P_a$, and so on. Therefore, if agent $a$ had preferences $P_a^h=\gamma_a^h\oplus \left.P_a\right|_{\Omega_{k}\backslash \{\omega_k\}}$, she would have the choices observed in $h^{A-P}_a$. Since $\mathbb{S}^{\varphi}$ sequentializes $\varphi$, the outcome of agent $a^*$ produced after the collective history $h^{A-P}$ is, therefore:

\begin{equation}
    \label{eq:lemma1ProofThm2-1}
    \left.\mathcal{O}_{a^*}\right|_{h^A}\left(\sigma_{a^*}^{P},\sigma_{-a^*}^{P}\right)=\varphi_{a^*}\left(\gamma_{a^*}^h\oplus \left.P_a\right|_{\Omega_{k}\backslash \{\omega_k\}}, P_{-a}^h\right)
\end{equation}

Suppose now that agent $a^*$ follows some arbitrary strategy $\sigma'$ while the others follows straightforward strategies with respect to $P$ after $h^A$. The resulting collective history is in general different from $h^{A-P}$, but for each agent $a$ it has the following format:

\[h^{A-P'}_a=h^A_a\oplus\left( \left(\Omega'_{k+1},\omega'_{k+1}\right), \left(\Omega'_{k+2},\omega'_{k+2}\right), \ldots, \left(\Omega'_{\ell'},\omega'_{\ell'}\right) \right)\]

For every agent except for $a^*$, the menus that they are presented after $h^A_a$---and also their choices as a result--- might be different from those in $h^{A-P'}_a$. But since by assumption they are all consistent with their preferences in $P$, it is also the case that if they had preferences $P_a^h=\gamma_a^h\oplus \left.P_a\right|_{\Omega_{k}\backslash \{\omega_k\}}$ they would have made the same choices. Agent $a^*$, on the other hand, followed different preferences, which resulted in a different $h^{A-P'}_{a^*}$, in which choices made after $h^A_a$ might not be the same as the ones she would make while following the straightforward strategy with respect to $P_{a^*}$. However, for the same reason highlighted above for $h^A_a$, we could construct a preference $P^*_a$ that rationalizes the choices in $h^{A-P'}_{a^*}$. The outcome of agent $a^*$ produced after the collective history $h^{A-P'}$ is, therefore:

\begin{equation}
    \label{eq:lemma1ProofThm2-2}
    \left.\mathcal{O}_{a^*}\right|_{h^A}\left(\sigma',\sigma_{-a^*}^{P}\right)=\varphi_{a^*}\left(\gamma_{a^*}^h\oplus \left.P^*_a\right|_{\Omega_{k}\backslash \{\omega_k\}}, P_{-a}^h\right)
\end{equation}

For $I=\{o\in\gamma_{a^*}^h\}$, equations \ref{eq:lemma1ProofThm2-1} and \ref{eq:lemma1ProofThm2-2} finish the proof.
\end{proof}

Remember that by construction $\mathbb{S}^{\varphi}$ sequentializes $\varphi$, and therefore when all agents use straightforward strategies with respect to their preferences $P$, $\varphi(P)$ is the outcome produced. Suppose now, for contradiction, that $\mathbb{S}^{\varphi}$ does not implement $\varphi$ in perfect ex-post equilibrium. Then, there is a preference profile $P^*\in\mathcal{P}$ for which a strategy profile in which all agents follow straightforward strategies is not a a perfect ex-post equilibrium. That is, there is a collective history $h^A\in H^A_{\mathbb{S}}$, an agent $a$ and a strategy $\sigma_a'$, which is not straightforward, for which $\left.\mathcal{O}_{a}\right|_{h^A}\left(\sigma_a',\sigma_{-a}^{P^*}\right)\ P_a^*\  \left.\mathcal{O}_{a^*}\right|_{h^A}\left(\sigma_a^{P^*},\sigma_{-a}^{P^*}\right)$, where $\sigma_{i}^{P^*}$ is the straightforward strategy for some agent $i$ with respect to $P^*$.

By Lemma \ref{lem:outcomeAfterDeviation}, when the preference profile is $P^*$ and all agents follow the straightforward strategies after $h^A$,  outcome for $a$ under the straightforward strategy and $\sigma_a'$ are, respectively:

\[o=\varphi_{a}\left(\gamma\oplus \left.P'_a\right|_{O\backslash I}, P^*_{-a}\right)\text{  and  }o'=\varphi_{a}\left(\gamma\oplus \left.P^*\right|_{O\backslash I}, P^*_{-a}\right)\]

for some $I\subseteq O$, a tuple $\gamma$ with the elements in $I$, and preference $P'_a$. Since $a$ obtains a better outcome, $o\ P_a^*\ o'$. But by strategy-proofness of $\varphi$ and corollary \ref{cor:monotoneSP}, $o'\ R_a\ o$. Contradiction. Therefore, straightforward strategies constitute an ex-post perfect equilibrium for the game induced by $\mathbb{S}^{\varphi}$. Therefore, if $\varphi$ satisfies monotonic discoverability and is strategy-proof, any pick-an-object mechanism that sequentializes $\varphi$ implements $\varphi$ in a perfect ex-post equilibrium in straightforward strategies.

Next, suppose that $\varphi$ satisfies monotonic discoverability but is not strategy-proof. Moreover, suppose for contradiction that there is a pick-an-object mechanism $\mathbb{S}^{\varphi}$ that implements $\varphi$ in perfect ex-post equilibrium. This implies that $\mathbb{S}^{\varphi}$ sequentializes $\varphi$ and that straightforward strategies constitute a perfect ex-post equilibrium. Since $\varphi$ is not strategy-proof, then there is preference profile $P^*$, an agent $a^*$ and a preference $P'_a\neq P^*_a$ such that $\varphi_a(P'_a,P^*_{-a}) P^*_a \varphi_a(P*_a,P^*_{-a})$. Since $\mathbb{S}^{\varphi}$ sequentializes $\varphi$, when players have the preference profile $P^*$ and follow straightforward strategies from $h^{A-\emptyset}$, agent $a$ obtains the outcome $\varphi_a(P*_a,P^*_{-a})$. If agent $a$ deviates at $h^{A-\emptyset}$ and follows a straightforward strategy that acts \textit{as if} she had the preference $P'_a$, her outcome would be $\varphi_a(P'_a,P^*_{-a})$, making that deviation profitable. This is a contradiction with straightforward strategies being a perfect ex-post equilibrium of the game induced by $\mathbb{S}^{\varphi}$.

Finally, if the rule is strategy-proof but does not satisfy monotonic discoverability, it cannot be implemented by a pick-an-object mechanism, since by Theorem \ref{thm:PickAnObjectIFFTerminalAv}, and the fact that strategy-proofness does not imply monotonic discoverability (see Example \ref{example:SPbutNotMonotDisc}) the rule cannot be sequentialized by a pick-an-object mechanism.

\subsection*{Proof of Propositions \ref{prop:PAOUniqueRM} and \ref{prop:DAUnique}}
\begin{proof}

The proof for both propositions rely on situations in which, starting from a period $t$, every agent who will make any choice at $t$ or later, except for $a^*\in A$, deems every remaining object unacceptable. For that, we will first show that if $\varphi$ is individually rational, resource monotonic and strategy-proof, $a^*$ will be matched to whatever object she chooses from her menu.


Let $\mathbb{S}$ be the canonical PAO mechanism that sequentializes $\varphi$---by monotonic discoverability and Theorem \ref{thm:PickAnObjectIFFTerminalAv} it exists. 

For the case where $\varphi$ is the Gale-Shapley DA, notice that the canonical PAO mechanism that sequentializes operates precisely as an instance of the Iterative Deferred Acceptance Mechanism introduced in \cite{Bo2016-id}, in which agents can choose only one option at a time: at $t=1$ each agent $a$ receives a menu with all the object-types (or colleges) in which she is deemed acceptable. For every $t>1$, every student who was rejected\footnote{That is, for whom the last object type she chose is no longer feasible.} in period $t-1$ is given a menu which includes $\emptyset$ and every object type that is still feasible for her. Notice that, for the Gale-Shapley DA rule, an object $o$ being feasible for an agent in period $t$ has two implications: (i) if she is the only agent choosing $o$ in period $t$, then she will be ``tentatively matched'' to $o$, instead of rejected, by the end of period $t$, and (ii) if no other agent chooses $o$ in any period after $t$, when the final allocation is produced she will be matched to $o$. In what follows, we will refer to these observations as \emph{Observation (*)}.

For the case considered in proposition \ref{prop:PAOUniqueRM}, the argument needs additional work. Let $h^{A}$ be a collective history of $\mathbb{S}$, where $\mathbb{S}(h^{A})=\left(\phi_a\right)_{a\in A}$ and at least one of these menus is non-empty.\footnote{That is, there is at least one agent with a non-empty menu to choose from after $h^{A}$.} Let $a^*\in A$ be an agent with one of these non-empty menus $\phi_{a^*}$. 


\begin{lem}\label{lem:OnlyOneChoosingGetsIt}
  Let $\varphi$ be individually rational, resource monotonic and strategy-proof. Let moreover $h^{A-o}$ be the continuation history of $h^{A}$ in which agent $a^*$ chooses object type $o$ from the menu $\phi_{a^*}$, and every other choice made in the same period or any following period, made by any agent, is $\emptyset$. Let $P^*$ be any preference profile consistent with $h^{A-o}$. Then, $\varphi_{a^*}(P^*)=o$.
\end{lem}

\begin{proof}
Let $P^A$ be any preference profile consistent with $h^{A}$. Let $\mu^A$ be an allocation in which every agent is matched to her last choice in $h^{A}$. Since $\mathbb{S}$ sequentializes $\varphi$, and $o\in \phi_{a^*}$, there is a preference profile $P^{A-o}$ that is a continuation profile of $P^A$ at $\mu^A$ for which $\varphi_{a^*}(P^{A-o})=o$.\footnote{That is, since $o$ is in the menu given to $a^*$, there must be a preference profile consistent with $h^{A}$ such that $a^*$ is matched to $o$ under $\varphi$. Otherwise, $o$ would not be in the menu.}

Now, let $P^{1}$ be the same as $P^{A-o}$ except that $o$ is moved to come right after $\mu^A(a^*)$ in $P^{A-o}_{a^*}$---if that was not already the case in $P^{A-o}_{a^*}$. By strategy-proofness of $\varphi$, $\varphi_{a^*}(P^{1})=o$.\footnote{If that was not the case, it would imply that $a^*$ could obtain $o$ by reporting it as less desirable than it actually is, which would contradict strategy-proofness.}

Next, let $P^2$ be the same as $P^1$ except that for every $a\neq a^*$, every object $o\in O\backslash\{\emptyset\}$ such that $\mu^A(a) P^1_a o$, $\emptyset\ P^2_a\ o$.\footnote{That is, every object below $\mu^A(a)$ is unacceptable.} Notice that $P^2$ is consistent with $h^{A-o}$.

Individual rationality of $\varphi$ implies that under $\varphi(P^2)$ agents other than $a^*$ are matched to either their match under $\mu^A$ or to $\emptyset$. By resource monotonicity of $\varphi$, $a^*$ cannot obtain in $\varphi(P^2)$ an outcome that is worse than $o$ with respect to $P^1$. By definition of PAO mechanism, $a^*$ cannot be matched to an object chosen previously. Since $\mathbb{S}$ sequentializes $\varphi$, therefore, this must imply that $\varphi_{a^*}(P^2)=o$.

Finally, notice that by construction $\left|\varphi(h^{A-o})\right|=1$,\footnote{That is, $h^{A-o}$ is defined to include choices from menus until an outcome is produced.} and that any $P^*$ consistent with $h^{A-o}$ results in the same outcome. Therefore, $\varphi(P^*)_{a^*}=\varphi(P^2)_{a^*}=o$.
\end{proof}

Lemma \ref{lem:OnlyOneChoosingGetsIt} implies, therefore, that while using a PAO mechanism that sequentializes a rule that is strategy-proof and resource monotonic, if starting at some period only one agent chooses an object, then that agent will keep that object---as is also the case with the Gale-Shapley DA rule. 

By Theorem \ref{thm:PAOTruthfullIFSP+MD}, straightforward strategies constitute a perfect ex-post equilibrium in the PAO mechanism $\mathbb{S}$. We will now prove by backward induction on the number of steps involved in that mechanism that no other type-strategy profile also constitutes a perfect ex-post equilibrium. 

Let $T^{MAX}$ be the maximum number of steps that can result from agents interacting with the mechanism.\footnote{Notice that since the set of object types is finite and menus given to agents never include their last choices, there is such finite $T^{MAX}$.} Choices made by agents in period $T^{MAX}-1$, therefore, are final. By definition of PAO mechanisms, this implies that every agent who makes a choice in period $T^{MAX}-1$ will be assigned the object in the menu that they chose. Therefore, any type-strategy profile in which an agent does not simply picks her most preferred object-type from a menu in period $T^{MAX}-1$ with respect to her true preference cannot be part of a perfect ex-post equilibrium.

Next, suppose that this is period $t<T^{MAX}-1$ while running the PAO mechanism $\mathbb{S}$, and the set of agents $A^*\subseteq A$ receive non-empty menus $\left(\phi_a\right)_{a\in A^*}$. By backward induction---which we will denote by \emph{backward induction assumption}---every choice made in a perfect ex-post equilibrium by an agent in every period after $t$ follow straightforward strategies. 

Suppose for contradiction, that there is an agent $a^*\in A^*$ who, with a preference $P_{a^*}$, does not follow a straightforward strategy---choosing $o'$ from a menu $\phi_{a^*}$, while $o^*$ is the most-preferred object in $\phi_{a^*}$ with respect to $P_{a^*}$.\footnote{In principle, the deviation could involve other changes in agent $a^*$'s choices after period $t$. The argument we follow, however, involves a situation in which her choice in period $t$ is her last one.}

Let $h^{A-t}=\left\{h_a\right\}_{a\in A}$ be the collective history followed all the players up to period $t$. Construct now a preference profile $P^{t*}$, where, for each $a\in A$:

\[h_a=\left((\Omega_1,\omega_1),(\Omega_2,\omega_2),\ldots,(\Omega_k,\omega_k)\right)\]
implies
\[\omega_1\ P^{t*}_a\ \omega_2\ P^{t*}_a \cdots \ P^{t*}_a\ \omega_k\ P^{t*}_a\ \emptyset \text{ and for every } o\in O\backslash \{\omega_1,\ldots,\omega_k,\emptyset\}\text{, }\emptyset\ P^{t*}_a o\]

Notice that any such $P^{t*}$ is trivially consistent with $h^{A-t}$.\footnote{Notice that since we are considering perfect ex-post equilibrium, we should not assume that any agent, including $a^*$, followed straightforward strategies before $t$.} Moreover, by backward induction assumption, under $P^{t*}$, every agent who is given a menu in any period after $t$ will choose the option $\emptyset$---since agents follow straightforward strategies after $t$, and every item in any menu that they receive is deemed unacceptable with respect to their preferences in $P^{t*}$.

By assumption, agent $a^*$ will choose the object $o'$ from her menu. If $a^*$ is the only agent making a choice in period $t$, Lemma \ref{lem:OnlyOneChoosingGetsIt} or \emph{Observation (*)} imply that this is the last choice that $a^*$ will make, and that she will be matched to $o'$. But Lemma \ref{lem:OnlyOneChoosingGetsIt} or \emph{Observation (*)} also imply that whatever object she chooses, she will be matched to it, including $o^*$, which is strictly preferred to $o'$ and any other item in her menu. Therefore, a strategy in which $a^*$ chooses any object other than $o^*$ when she is the only player making a choice at period $t$ cannot be part of a perfect ex-post equilibrium when her preference is $P_{a^*}$.

It must be, therefore, that there are other agents making choices from menus in period $t$, that is, $|A^*|>1$. Take any agent $a^{\emptyset}\in A^*\backslash\{a^*\}$. Since her preferences in $P^{t*}$ imply that every object in the menu she receives is unacceptable, no strategy in which she is not left unmatched can be part of a perfect ex-post equilibrium for this preference.

We now separate the remaining arguments depending on the assumptions about $\varphi$.

\textbf{Case 1: $\varphi$ is the Gale-Shapley DA rule.}

If every agent $a^{\emptyset}\in A^*\backslash\{a^*\}$ chooses $\emptyset$, then by backward induction assumption and \emph{Observation (*)} , agent $a^*$ will be matched to $o'$, but would end up matched to $o^*$ if she chose it, which is a contradiction with this strategy being part of a perfect ex-post equilibrium when her preference is $P_{a^*}$. 

Next, if any agent $a^{\emptyset}\in A^*\backslash\{a^*\}$ chooses an object type and is not rejected, by backward induction assumption and \emph{Observation (*)}, she will by the end of the procedure be matched to an unacceptable outcome, while if she chose $\emptyset$ instead, she would be strictly better off. Therefore, no agent other than $a^*$ can choose an object type after $h^{A-t}$ and not have that choice rejected by the end of that period in a perfect ex-post equilibrium strategy profile. Notice, moreover, that given the backward induction assumption, if more than one agent chooses some object type in period $t$, at least one of them will end up matched to it by the end of the procedure, which again cannot be part of a perfect ex-post equilibrium.

Therefore, the only possibility that remains is one where every agent who chooses an object type other than $\emptyset$ in period $t$ also chooses $o'$, but all of them have their choice rejected by the end of that period, and therefore by backward induction assumption will end up unmatched. Notice, however, that $o'$ was feasible for $a^*$, she chose $o'$ in period $t$ and every other agent who chose $o'$ in the same period was rejected. This implies that (i) $a^*$ is tentatively matched to $o'$ by the end of period $t$, and by backward induction assumption will end up matched to it by the end of the procedure, and (ii) if she chose $o^*$, she would be the only agent doing so in period $t$, and by \emph{Observation (*)} would end up matched to $o^*$. Choosing any object other than $o^*$, therefore, cannot be part of a perfect ex-post equilibrium, contradicting the assumption that she chooses $o'$ under some equilibrium.

\textbf{Case 2: $\varphi$ is individually rational, resource monotonic, strategy-proof and satisfies monotonic discoverability.}


If $a^{\emptyset}$ and every player other than $a^*$ who makes a choice in period $t$ chooses $\emptyset$, Lemma \ref{lem:OnlyOneChoosingGetsIt} would yield a contradiction, since $a^*$ has a strategy that would result in a strictly better outcome---choosing $o^*$. Suppose, then, that there is a perfect ex-post equilibrium in which some set $A'\subset A^*\backslash\{a^*\}$ of these players do not choose $\emptyset$ in period $t$, but instead choose some other object types at $t$, but still end up left unmatched when the procedure ends.\footnote{This could happen if, for example, these agents are not matched to the object type they chose when $a^*$ chooses some object other than $\emptyset$.} The choices that result from these non-straightforward choices from the agents in $A'$, combined with the choices made by other players and the backward induction assumption, are consistent with a preference profile where the preferences of all agents in $A\backslash(A'\cup \{a^*\})$ are the same as in $P^{t*}$, and agents in $A'$ have preferences different than those in $P^{t*}$. Denote, for each agent $a\in A'$, these different preferences as $P'_a$.

Let $o$ be any object type in the menu faced by agent $a^*$ in period $t$. Let $P_{a^*}^{o}$ be a preference for $a^*$ that is exactly as $P_{a^*}^{t*}$ except that $o$ is the least-preferred acceptable object.\footnote{That is, it differs from $P_{a^*}^{t*}$ only in that the object $o$ is moved to the position right above $\emptyset$ in the ranking.} By Lemma \ref{lem:OnlyOneChoosingGetsIt}, the definition of $P^{t*}$, and the fact that the $\mathbb{S}$ sequentializes $\varphi$, $\varphi_{a^*}(P_{a^*}^{o},P_{-a^*}^{t*})=o$. Since $\varphi$ is individually rational, for each $a\in A'$, $\varphi_{a}(P_{a^*}^{o},P_{-a^*}^{t*})=\emptyset$.\footnote{Remember that under $P^{t*}$ all object types in the menus given to agents in $A'$ are unacceptable.} Next, let for an arbitrary $a'\in A'$,  $P_{-a^*}^1$ be the same as $P_{-a^*}^{t*}$ except that $P_{a'}^1=P'_{a'}$. Since $\varphi$ is resource monotonic and $a'$ is left unmatched under both $\varphi(P_{a^*}^{o},P_{-a^*}^{t*})$ and $\varphi(P_{a^*}^{o},P_{-a^*}^{1})$, $a^*$ cannot be worse off and $\varphi(P_{a^*}^{o},P_{-a^*}^{t*})_{a^*}=\varphi(P_{a^*}^{o},P_{-a^*}^{1})_{a^*}$. If we keep replacing the preferences of each agent $a\in A'$ with $P'_a$, resource monotonicity implies that the outcome obtained by $a^*$ does not change. Denote the preference profile that results from this entire process by $P^{*-o}$. We conclude, therefore, that $\varphi_{a^*}(P^{*-o})=o$.

In any perfect ex-post equilibrium, the strategies used by the players in $A^*\backslash\{a^*\}$ cannot result in them being matched to any unacceptable object. Therefore, no matter what strategy they follow from period $t$ on, we can construct a preference profile $P^{*-o}$ as above that is consistent with the resulting collective history.

Notice, now, that the profile $P^{*-o}$ is consistent with (i) an arbitrary set of agents $A'$ making choices at time $t$ which, despite not being straightforward, still results in the unique (conditional) acceptable outcome, which is remaining unmatched, (ii) the remaining players following straightforward strategies, and (iii) player $a^*$ choosing an arbitrary object $o$ in the menu she is given in period $t$. 

Since $\mathbb{S}$ sequentializes $\varphi$, we have therefore that in any perfect ex-post equilibrium $a^*$ is, therefore, matched to the object she chooses in the menu she is given in period $t$ under the preference profile $P^{t*}$. Choosing any object other than $o^*$ cannot be part of a perfect ex-post equilibrium, contradicting the assumption that she chooses $o'$ under some equilibrium.

\end{proof}

\begin{customthm}{\ref{prop:OSPImplementableViaPAO}}
Every non-bossy OSP implementable rule is pick-an-object implementable in weakly dominant strategies.
\end{customthm}

\begin{proof}
We prove the result by constructing a pick-an-object mechanism which serves as an interface between the agents and the millipede game of the OSP mechanism. Let $\varphi^{OSP}$ be an OSP-implementable rule. Then, by \cite{Pycia2019-vl}, there is a millipede game $\Gamma$ where, for every preference profile $P$, each agent following a greedy strategy with respect to their preference, results in the allocation $\varphi^{OSP}(P)$. In the game $\Gamma$, Nature moves once, at history $h^{\emptyset}$\footnote{If Nature does not move, then we can simply consider that Nature chooses a constant action at $h^{\emptyset}$.}, and all players have perfect information. For every non-terminal history $h$ (except for the one in which Nature moves,) there is an associated player $P(h)\in A$ and a menu $\phi(h)$.

By the definition of the millipede game in \cite{Pycia2019-vl}, the menu $\phi(h)$ might contain multiple items, each of them associated with an element of $O$, and at most one $pass$ option. In principle, there might be multiple items associates with the same $o\in O$, and the choice among these options be consequential to the final allocation for other players. The characterization of OSP mechanisms as milipede games associated with a greedy strategy allows for that selection to depend on the preferences of the player making that choice. For example, you could have an agent $a$ who receives a menu $\phi(h)$ containing, among other options, items $\omega$ and $\omega'$. Choosing either will result in $a$ being assigned object $o$, but the greedy strategy associated with this millipede game indicates that the agent chooses $\omega$ if her preference is $o P_a o' P_a o''$, but $\omega'$ if her preference is $o P_a o''P_a o'$. 

Non-bossiness, however, implies that even if $\Gamma$ contains multiple items for the same private allocation that are chosen, under the greedy strategy, when the agent has different preferences, these must result in the same allocation, fixed the other agents' greedy strategies. Therefore, every item associated with a private allocation is equivalent to each other. This allows us to assume, instead that items in the menu given to the players in a millipede game consist are all objects, with at most one pass option. That is, $\phi(h)\subseteq O\cup\{pass\}$.

With that, we can define the set of clinchable objects for agent $a$ at continuations of $h$ as follows:

\[C_a(h)=\bigcup_{h'\in P^{-1}(a):h\subseteq h'} \phi(h')\backslash \{pass\}\]

That is, $C_a(h)$ is the set of objects that are given in a menu to player $a$ at some continuation history from $h$. The menu function $\mathbb{S}$ for the pick-an-object mechanism that we are constructing is such that $\mathbb{S}\left(h^{A-\emptyset}\right)=\left(C_{a_1}\left(h^{\emptyset}\right),C_{a_2}\left(h^{\emptyset}\right),\ldots,C_{a_n}\left(h^{\emptyset}\right)\right)$. The value of $\mathbb{S}\left(h^{A}\right)$ for any collective history $h^{A}=\left(h_1,h_2,\ldots,h_n\right)$ is defined as follows:

\begin{itemize}
    \item Let $h_i^j$ be the $j$th item in $h_i$, and $h_i^j=\left(\Omega_j^i,\omega_j^i\right)$. That is, $\Omega_j^i$ is the $j$th menu faced by agent $a_i$, in which she chose the item $\omega_j^i\in \Omega_j^i$.
    \item Let $h=h^\emptyset$ and for all $i=1,\ldots,n$, let $\eta_i=1$. Follow the game tree in $\Gamma$ by making the agents play as follows:
    \begin{itemize}
        \item \textbf{Step 0}: If $h$ is a terminal node of $\Gamma$, then we can determine the value of $\mathbb{S}\left(h^{A}\right)$ to be $\left(\emptyset,\emptyset,\ldots,\emptyset\right)$.
        \item \textbf{Step 1}: Let $a_i=P(h)$. That is, $a_i$ is the active player at history $h$. If  $\omega_{\eta_i}^i\not\in C_{a_i}(h)$ and $\eta_i<|h_i|$, increase the value of $\eta_i$ by $1$. Otherwise,  if  $\omega_{\eta_i}^i\not\in C_{a_i}(h)$ and $\eta_i=|h_i|$ we can determine the value of $\mathbb{S}\left(h^{A}\right)$ to be $\left(\phi_1,\phi_2,\ldots,\phi_n\right)$, where $\phi_j=\emptyset$ if $\omega_{\eta_j}^j\in C_{a_j}(h)$, and $\phi_j=C_{a_j}(h)$ otherwise.
        \item \textbf{Step 2}: If $\omega_{\eta_i}^i\in \phi(h)$, follow the node that represents choosing $\omega_{\eta_i}^i$, and let $h$ be the history that follows that choice. Otherwise, choose the node $Pass$, and let $h$ be the history that follows that choice. In either case, go back to step 0.
    \end{itemize}
\end{itemize}

Notice that the procedure above will produce a list of menus for any collective history, even those which could never take place under a pick-an-object mechanism. To get our result, however, it suffices for us to show that collective histories that are generated by agents following straightforward strategies while interacting with the pick-an-object mechanism $\mathbb{S}$ described above are translated to greedy strategies from these same agents playing the game $\Gamma$.

Suppose then, for contradiction, that agents follow straightforward strategies while interacting with the pick-an-object mechanism $\mathbb{S}$ but that is not translated, in the description above, into these agents following greedy strategies. That must imply, therefore, that there is an agent $a^*\in A$ who in the procedure described above either (i) chooses $Pass$ at some history $h^*$, and the most-preferred element in $C_{a^*}\left(h^*\right)$ with respect to $P_{a^*}$ is in $\phi\left(h^*\right)$, or (ii) chooses an object $o\in \phi\left(h^*\right)$ while there is another object $o'\in C_{a^*}\left(h^*\right)$ for which $o' P_{a^*} o$.

Fist, consider case (i). Notice that, by the definition of the value of $\mathbb{S}\left(h^{A}\right)$ and by the description of Step 1, every time an agent receives a menu of objects, that menu contains all objects in the continuation outcomes of the history that is reached by following the collective history that precedes the offering of that menu. Moreover, Step 1 also implies that whenever a history $h$ in which the last choice made by $a^*$ from a menu is not in $C_{a^*}\left(h\right)$, the procedure does not follow the game $\Gamma$ after $h$. Therefore, at every history $h$, the last choice made by $a^*$ in her choice history is her most-preferred object in $C_{a^*}(h)$ with respect to $P_{a^*}$. Step 2 requires that $Pass$ would only be chosen if that last choice was \textbf{not} in $\phi\left(h^*\right)$, a contradiction.

Next, case (ii). In the previous paragraph we established that every time an agent receives a menu of objects, it contains all objects in the continuation outcomes. Since the agent chooses $o$ in the game $\Gamma$, Step 2 and the way in which the pick-an-object mechanism works requires that $o$ must have been the last choice made by $a^*$. Since both $o$ and $o'$ are in $C_{a^*}\left(h^*\right)$, they both must be in the menu from which $a^*$ chose $o$. But that means that $o' P_{a^*} o$ and $a^*$ chose $o$ from a menu containing $o'$, which is a contradiction with $a^*$ following the straightforward strategy.

Finally, note that any deviation from straightforward strategies in the pick-an-object mechanism $\mathbb{S}$ either leads to the same paths through $\Gamma$ as the greedy strategy or to different ones. Since the greedy strategy is obviously dominant in $\Gamma$, then it weakly dominates any other strategy, including any deviation induced by deviations from straightforward strategies in $\mathbb{S}$. Therefore, straightforward strategies are weakly dominant in the pick-an-object mechanism $\mathbb{S}$.\footnote{Notice, however, that straightforward strategies may not be obviously dominant in the game induced by the pick-an-object mechanism $\mathbb{S}$. One can easily see, for example, that in the pick-an-object implementation of serial dictatorship used in our experimental section, the worst outcome that could come from following a straightforward strategy is typically worse than the best outcome that can come from a deviation from it in the first step of the mechanism.}
\end{proof}

\section{Additional experimental results}

\subsubsection{Determinants of truthful behavior\label{sec:App_truthful}}

Here, we take a closer look at the determinants and the nature of deviations from the truthful strategies.

\noindent \textbf{Determinants of the truthful strategy:}
\begin{enumerate}
    \item Under the TTC rule with cyclic priorities, there is no significant correlation between the propensity to play the truthful strategy and how high the participant is in the priorities of objects in Direct, but there is significant correlation in PAO. 
    \item Under the TTC rule with acyclic priorities, there is a significant positive correlation between how high the participants are in the priorities of objects and the propensity to play the truthful strategy in all treatments. 
    \item Under the SD rule, the manipulations in Direct and PAO mechanism are significantly more likely to take place if participants had a lower score.
\end{enumerate}
\medskip{}

\noindent \textbf{Support:} 

\begin{table}[]
    \centering\smaller
    
\begin{tabular}{l*{6}{c}}
\hline\hline
                    &\multicolumn{1}{c}{(1)}&\multicolumn{1}{c}{(2)}&\multicolumn{1}{c}{(3)}&\multicolumn{1}{c}{(4)}&\multicolumn{1}{c}{(5)}&\multicolumn{1}{c}{(6)}\\
                    &\multicolumn{1}{c}{Direct TTC}&\multicolumn{1}{c}{PAO TTC}&\multicolumn{1}{c}{Direct TCC}&\multicolumn{1}{c}{PAO TTC}&\multicolumn{1}{c}{Direct }&\multicolumn{1}{c}{PAO}\\
                    &        cyclic         &        cyclic         &       acyclic         &        acyclic         &         SD         &         SD         \\
\hline
          &                     &                     &                     &                     &                     &                     \\
Rank at best obj.   &       -.015         &       -.049{***}&       -.088{***}&       -.117{***}&                     &                     \\
                    &      (.010)         &      (.008)         &      (.014)         &      (.015)         &                     &                     \\

Rank at 2nd-best obj.          &        .003         &        .026{***}&        .007         &        .011         &                     &                     \\
                    &      (.011)         &      (.008)         &      (.012)         &      (.007)         &                     &                     \\

Av. rank at rem. obj.&            .026         &        .039{**} &        .053{**} &        .053{***}&                     &                     \\
                    &      (.028)         &      (.015)         &      (.023)         &      (.016)         &                     &                     \\

Score in SD               &                         &                     &                     &                     &        .006{***}&        .004{***}\\
                    &                     &                     &                     &                     &      (.001)         &      (.001)         \\

\hline
Observations        &          385         &         623         &         371         &         532         &         462         &         231           \\
No. of clusters  &              12       &          12           &          12           &      12               & 12              &    12                 \\
log(likelihood)     &     -265.56         &     -413.00         &     -238.66         &     -308.94         &     -281.28         &     -155.20         \\

\hline\hline
\end{tabular}

\begin{tablenotes}
\item \begin{flushleft}
Marginal effects of Probit regressions of the truthful strategy. The sample excludes participants who either always played the truthful strategy or never played the truthful strategy  within a treatment. Rank at best obj.  is the rank of the participant in the priority order of her favorite objects (goes from 1 to 8). Rank at 2nd-best obj. is the rank of the participant in the priority order of her second favorite objects (goes from 1 to 8). Av. rank at rem. obj. is the average rank of the participant in priorities of all remaining objects. {*}{*}$p<0.05$, {*}{*}{*} $p<0.01$. Standard errors are clustered at the level of matching groups and are presented in parentheses. 
\end{flushleft}
\end{tablenotes}
   \caption{Marginal effects of the probit model of truthful strategies on the priority ranks in TTC and the score in SD for Direct and PAO treatments}
    \label{tab:truthpriority}
\end{table}

Table \ref{tab:truthpriority} shows the marginal effects of probit regression models for the dummy of playing the truthful strategy in Direct and PAO, depending on the priority of the participants. Each model of Table \ref{tab:truthpriority} restricts the sample to one of the environment-treatment combination. 

In all environments and treatments---except Direct TTC with cyclic priorities---the priorities of the objects can partially explain deviations from the truthful strategy. The higher the priority in the most-preferred object (i.e., the lower the rank), the more likely it is that the participant will play the truthful strategy. The rank in other objects has the opposite average effect: the lower the priority in other objects, the more likely it is that the participants will play the truthful strategy. This result suggests that subjects attempt to list the object where they have relatively high priority at the top of the rank-ordered list or escape ranking highly the objects where they have very low priority. Note that this is a typical deviation also in SD environments for both Direct and PAO treatments, where participants with a low score misrepresent their preferences more often than participants with a high score. Interestingly, the priorities of objects are significant in explaining deviations from truthful strategies in the Direct treatment only under TTC with acyclic priorities, and not under cyclic priorities.

\subsubsection{Alternative efficiency definitions\label{sec:App_eff}}

Given the fact that it might be the case that the equilibrium allocation has a lower sum of ranks than the allocation when a participant deviates from the truthful strategy, it is essential to look at different criteria. One measure of the success of the mechanism is whether the desired allocation is reached. This approach is used in \citet{Li2017-oa} to estimate the performance of direct versus OSP SD. We adopt this approach for all treatments. Note, however, that in  \cite{Li2017-oa}, the market consists of four participants, while in our case, it consists of eight participants. In the larger market, we can expect a lower rate of the equilibrium allocations reached, as it is enough for one participant in the group to deviate from truthful strategy in the consequential way to distort the whole allocation.

\begin{table}[h]
    \centering 
    \begin{tabular}{|c|c|c|c|c|c|c|}
\hline 
 & Direct & PAO & OSP & Direc=PAO & Direct=OSP & PAO=OSP\tabularnewline
 &  &  &  & p-value & p-value & p-value\tabularnewline
\hline 
TTC cyclic & 5.9\% & 5.9\% & n/a & 1.00 & n/a & n/a\tabularnewline
\hline 
TTC acyclic & 22.6\% & 26.2\% & 12.6\%& 0.72 & 0.05 & 0.01\tabularnewline
\hline 
SD & 38.1\% & 50\% & 83.\% & 0.16 & 0.00 & 0.00\tabularnewline
\hline 
\end{tabular}
\begin{tablenotes}
\footnotesize{}
\item \emph{Notes:} All the p-values are for the two-sided Fisher exact test for equality of proportions of the equilibrium strategy allocations by treatments. test is performed on the allocation-level data
\end{tablenotes}
    \caption{Proportions of equilibrium allocations reached.}
    \label{tab:binaryalloc}
\end{table}

Table \ref{tab:binaryalloc} presents the proportion of equilibrium allocations by treatments. Under TTC with cyclic priorities, only 5.9\% of allocations in Direct and PAO were equilibrium allocations. This low rate is not surprising, given the rate of the truthful strategies manipulations, and the fact that even one consequential manipulation distorts the allocations. Under TTC with acyclic priorities, we observe that the OSP treatment has the lowest rate of equilibrium allocations. Despite the higher rate of truthful strategies, as every deviation from the truthful strategy more likely to be consequential for the allocation, it leads to only 12.6\% of equilibrium allocations on average, which is significantly lower than in Direct and PAO treatments. Finally, under SD, OSP outperforms Direct and PAO mechanisms. Thus, we replicate the finding by \citet{Li2017-oa}.  

While the comparison of the rate of equilibrium allocations is useful for a smaller market, or under very low deviations from truthful strategies, we think the criterion is not very informative about the consequences of deviations in the case of a high rate of deviations from truthful strategies. Another approach would be to analyze the difference in the consequence of deviation from truthful strategies from an individual perspective by treatments. We define cost of deviation as the average difference between the payoffs of those who played truthfully and those who deviated from truthfulness, controlling for the role in the market. 

\noindent \textbf{Result 4 (Cost of deviation from truthful strategy):}
\begin{enumerate}
    \item Under the TTC rule with cyclic priorities the average cost of deviation from the truthful strategy under Direct is 3.97 euros, under PAO 3.61 euros, with no significant difference between treatments.

    \item Under the TTC rule with acyclic priorities the average cost of deviation from the truthful strategy in Direct is 2.43 euros, in PAO 1.98 euros, and in OSP 3.36 euros with no significant difference between treatments.
    \item Under the SD rule  the average cost of deviation from the truthful strategy in Direct is 2.53 euros, in PAO 1.55 euros, and in OSP 6.01 euros with all differences between treatments being statistically significant.

\end{enumerate}

\noindent \textbf{Support:} 

Table \ref{costofmisrep} presents the results of the OLS estimation of the effect of misreporting on the payoff of the subjects.

Each regression includes 56 dummies for each combination of ID and round, to account for the ``role-specific'' fixed effects, as the roles (a combination of preferences and priorities/scores) vary the prospects of earning high payoffs. Thus, the coefficient for the non-truthful dummy presents the average differences between subjects who play truthfully relative to subjects who play non-truthfully in the Direct mechanism, controlling for the role of the subjects. Under TTC with cyclic priorities (Model (1) of Table \ref{costofmisrep}), the deviation from the truthful strategy, on average, leads to a loss of 3.97 euros in Direct (note that the maximum payoff for the allocation is 22 euros), while the deviations are 38 cents less costly in PAO, though the difference is not significant. Under TTC with acyclic priorities (Model (2) of Table \ref{costofmisrep}), the deviation from the truthful strategy, on average, leads to a loss of 2.43 euros in Direct, while the deviations are 45 cents less costly in PAO, and 93 cents more costly in OSP, though the differences are not significant. Finally, under SD, (Model (3) of Table \ref{costofmisrep}), the deviation from the truthful strategy, on average, leads to a loss of 2.53 euros in Direct, while the deviations are 98 cents less costly in PAO, and 3.48 euros more costly in OSP, with all differences being significant. The difference can be explained by the fact that skipping in PAO is less consequential, due to intermediate updates on which objects are left for allocation, which is not the case in Direct. The highest cost of deviations in OSP can be explained by the fact that all deviations from the truthful strategy are payoff-relevant in SD, as it is a unique equilibrium strategy in OSP, unlike in the other treatments.

\begin{table}[]
\centering{\footnotesize{}{
}%
\begin{tabular}{l*{3}{c}}
\hline\hline
                    &\multicolumn{1}{c}{(1)}&\multicolumn{1}{c}{(2)}&\multicolumn{1}{c}{(3)}\\
                    &\multicolumn{1}{c}{Payoff}&\multicolumn{1}{c}{Payoff}&\multicolumn{1}{c}{Payoff}\\
                    &        TTC cyclic         &        TTC acyclic         &        SD         \\
\hline
Non-truthful strategy       &      -3.97{***}&      -2.43{***}&      -2.53{***}\\
                    &      (.29)         &      (.29)         &      (.24)         \\
Non-truthful in PAO   &        .38         &        .45         &        .98{**} \\
                    &      (.39)         &      (.38)         &      (.36)         \\
Non-truthful in OSP   &                &       -.93         &      -3.48{***}\\
                    &                 &      (.60)         &      (.76)         \\
Dummies for each participant ID in each round&      yes               &       yes              &      yes\\
           
\hline
Observations        &        1344         &        2800         &        2072         \\
No. of clusters  &        24             &          37           &            37         \\
R$^2$               &        .287         &        .370         &        .759         \\
log(likelihood)     &    -4063.06         &    -7796.13         &    -5136.51         \\
\hline\hline
\end{tabular}{\footnotesize{}} }{\footnotesize\par}

\begin{flushleft}
{\footnotesize{}{}Notes:{\footnotesize{} }{\small{}OLS
regression. {*} $p<0.10$, {*}{*} $p<0.05$, {*}{*}{*} $p<0.01$.
Standard errors are clustered at the level of matching groups and
are presented in parentheses. Non-truthful is a dummy for not playing
the truthful strategy. Non-truthful in PAO is the interaction
of the Non-truthful dummy and the dummy for PAO treatment. Non-truthful
in OSP is the interaction of the Non-truthful dummy and the dummy
for OSP treatment.}{\footnotesize{} } }{\footnotesize\par}\end{flushleft}
 \caption{\emph{OLS regression of payoff on the dummy for non-truthful strategy
\label{costofmisrep}}}
\end{table}

\end{document}